\documentclass[pra,twocolumn,floatfix,superscriptaddress,longbibliography,notitlepage]{revtex4-1}

\usepackage{graphicx, color, graphpap}
\usepackage{enumitem}
\usepackage{amssymb}
\usepackage{amsthm}
\usepackage{multirow}
\usepackage[colorlinks=true,citecolor=blue,linkcolor=magenta]{hyperref}
\usepackage[T1]{fontenc}
\usepackage{bbm}
\usepackage{thmtools,thm-restate}
\usepackage{verbatim}
\usepackage{mathtools}
\usepackage{titlesec}
\usepackage{amsmath}
\usepackage[title]{appendix}
\usepackage[linesnumbered,ruled,vlined]{algorithm2e}
\SetKwInput{kwInit}{Init}

\long\def\ca#1\cb{} 

\newcommand{\braket}[2]{\langle #1 \hspace{1pt} | \hspace{1pt} #2 \rangle}

\newcommand{\bramatket}[3]{\langle #1 \hspace{1pt} | #2 | \hspace{1pt} #3 \rangle}

\newcommand{\norm}[2][]{#1| \! #1| #2 #1| \! #1|}

\newcommand{\ket}[1]{|#1\rangle}               
\newcommand{\bra}[1]{\langle #1|}              
\newcommand{\dya}[1]{\ket{#1}\!\bra{#1}}

\newcommand{\HC}{\mathcal{H}}

\newcommand{\PC}{\mathcal{P}}

\newcommand{\TC}{\mathcal{T}}

\newcommand{\Tr}{{\rm Tr}}
\newcommand{\sig}{{\rm sig}}

\renewcommand{\geq}{\geqslant}
\renewcommand{\leq}{\leqslant}

\DeclareMathOperator*{\argmin}{arg\,min}
\renewcommand{\vec}[1]{\boldsymbol{#1}}  

\newcommand{\ad}{^\dagger}

\newcommand{\tout}{{\text{out}}}
\newcommand{\tin}{{\text{in}}}

\newcommand{\thv}{\vec{\theta}}

{}
{}
\newtheorem{theorem}{Theorem}

\newtheorem{proposition}{Proposition}

\usepackage{amssymb}
\usepackage{dsfont}
\usepackage{bbold}

\usepackage[export]{adjustbox}
\theoremstyle{definition}

\usepackage{float}
\usepackage{mathrsfs}
\usepackage{rotating}
\usepackage[linesnumbered,ruled,vlined]{algorithm2e}
\SetKwInput{kwInit}{Init}

\begin{document}

\title{Entangled Datasets for Quantum Machine Learning}

\author{Louis Schatzki}
\email{louisms2@illinois.edu}
\affiliation{Theoretical Division, Los Alamos National Laboratory, Los Alamos, NM 87545, USA}
\affiliation{Department of Electrical and Computer Engineering, University of Illinois at Urbana-Champaign, Urbana, IL 61801, USA}

\author{Andrew Arrasmith}
\affiliation{Theoretical Division, Los Alamos National Laboratory, Los Alamos, NM 87545, USA}

\author{Patrick J. Coles}
\affiliation{Theoretical Division, Los Alamos National Laboratory, Los Alamos, NM 87545, USA}

\author{M. Cerezo}
\email{cerezo@lanl.gov}
\affiliation{Information Sciences, Los Alamos National Laboratory, Los Alamos, NM 87545, USA}
\affiliation{Theoretical Division, Los Alamos National Laboratory, Los Alamos, NM 87545, USA}
\affiliation{Center for Nonlinear Studies, Los Alamos National Laboratory, Los Alamos, NM, USA}

\begin{abstract}
High-quality, large-scale datasets have played a crucial role in the development and success of classical machine learning.  Quantum Machine Learning (QML) is a new field that aims to use quantum computers for data analysis, with the hope of obtaining a quantum advantage of some sort. While most proposed  QML architectures are benchmarked using classical datasets, there is still doubt whether QML on classical datasets will achieve such an advantage. In this work, we argue that one should instead employ quantum datasets composed of quantum states. For this purpose, we introduce the NTangled dataset composed of quantum states with different amounts and types of multipartite entanglement. We first show how a quantum neural network can be trained to generate the states in the NTangled dataset. Then, we use the NTangled dataset to benchmark QML models for supervised learning classification tasks. We also consider an alternative entanglement-based dataset, which is scalable  and is composed of states prepared by quantum circuits with different depths. As a byproduct of our results, we introduce a novel method for generating multipartite entangled states, providing a use-case of quantum neural networks for quantum entanglement theory.

\end{abstract}

\maketitle

\section{Introduction}

The field of Machine Learning (ML) revolutionized the way we use computers to solve problems. At their core, ML algorithms solve tasks without being explicit programmed to do so, but rather by learning from data and generalizing to previously unseen cases~\cite{mohri2018foundations}. Nowadays, ML is widely used and is considered a fundamental tool in virtually all areas of modern research and technology.

The historical development of ML owes a large part of its success to the progress in algorithmic techniques such as multi-perceptron neural networks~\cite{mohri2018foundations}, the backpropagation method \cite{Rumelhart1986},  support-vector machines~\cite{siegelmann1995computational}, and more generally kernel methods~\cite{hofmann2008kernel}. However, perhaps just as relevant as these techniques is the development of datasets that could be use to benchmark and improve ML architectures. In fact, it has been noted that some of the major ML breakthroughs such as  human-level spontaneous speech recognition, IBM's Deep Blue chess engine victory over Garry Kasparov, and  Google’s GoogLeNet software for object classification, were facilitated by access to high-quality datasets~\cite{wissner2016data}. Today, datasets are a staple in ML, and researchers benchmarking novel models can readily test their architectures by accessing a wide range of datasets such as  MNIST~\cite{lecun1998mnist}, Dogs vs. Cats~\cite{kaggle2016dogs}, ProteinNet~\cite{alquraishi2019proteinnet}, and Iris~\cite{andrews1985iris}.

The current availability of moderate-size quantum computers~\cite{preskill2018quantum} has led to an increasing interest in using these devices to achieve a quantum advantage, i.e., to perform a computational task faster than any classical supercomputer. In this context, Quantum Machine Learning  (QML)~\cite{schuld2015introduction,biamonte2017quantum,schuld2019quantum}, which generalizes the concept of classical ML based on the principles of quantum mechanics, is one of the most promising emergent applications to make practical use of quantum computers.

Despite tremendous recent progress~\cite{abbas2020power,huang2021provably,huang2021power,kubler2021inductive,cerezo2020variationalreview,kawai2020predicting,zhu2019training,killoran2019continuous}, the field of QML is still in its infancy. For instance, while certain architectures for Quantum Neural Networks (QNNs) have been developed~\cite{siomau2014quantum,tacchino2019artificial,beer2020training,romero2017quantum,farhi2018classification,abbas2020power,cong2019quantum}, there is still no consensus regarding which one should be favored. In fact, proposing novel QNNs is a difficult task due to the counter-intuitive nature of quantum mechanics. Thus, the search for scalable and efficient QNNs models is still an active field of research. Here, it is worth noting that  when a new QNN is introduced,  its success is usually tested through heuristics and benchmarking experiments.  

\begin{figure}[t]
\includegraphics[width=1\linewidth]{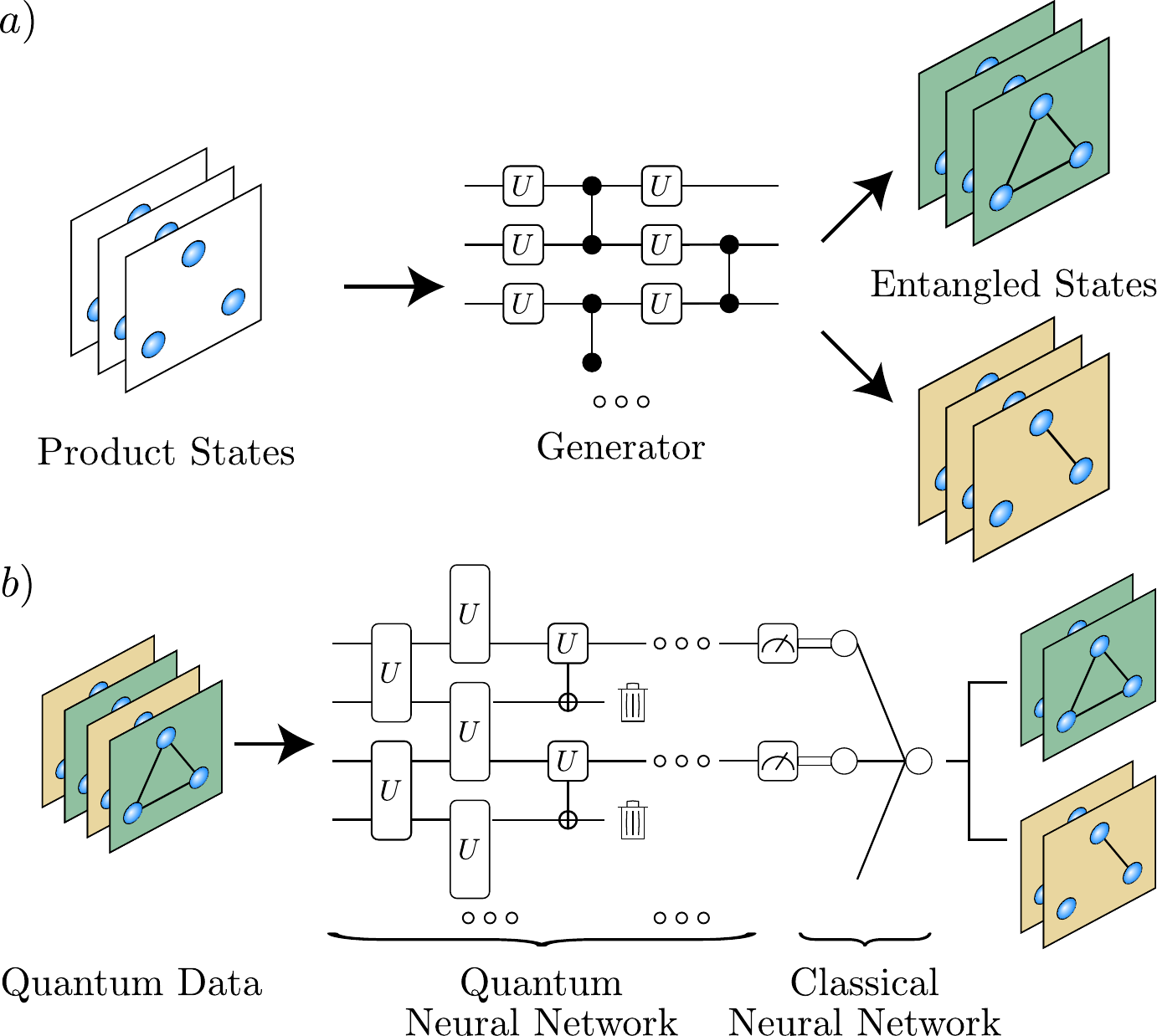}
\caption{\textbf{Overview of our main results}. (a) The states in our NTangled dataset are generated by training a Quantum Neural Network (QNN). Here, product states are inputted into a variational circuit trained so that the outputted states have a desired entanglement value. By changing the goal value, sets of states with varying amounts of entanglement are generated (represented with different colors). (b) The states in the NTangled dataset are used to benchmark a QML model for the supervised learning task of classifying the states according to their amount of entanglement. In the figure, the QML model is composed of a QNN, plus a one node classical neural network. }
\label{fig:overview}
\end{figure}

Currently, most researchers benchmark QML models using the same datasets employed in classical ML. To use these classical datasets in a quantum setting, one must encode the classical information into quantum states with an embedding scheme~\cite{havlivcek2019supervised,benedetti2019generative,lloyd2020quantum,schuld2020circuit,larose2020robust,huang2021power}. Thus, here the QML model is composed of the embedding (mapping classical information to quantum states) plus the trainable QNN (performing operations on the quantum states).  This leads to several issues. First, there is no general embedding method one can use, and the choice of embedding can affect the QML model trainability~\cite{thanaslip2021subtleties}. Moreover, recent results have raised doubts on the potential for achieving a  quantum advantage using classical data~\cite{kubler2021inductive}. Hence, there is the question of why one would benchmark QML performance on datasets that might not lead to quantum advantage.

This motivates the consideration of quantum datasets. Here, the data is composed of states generated from some quantum process~\cite{ uvarov2020machine,banchi2021generalization,bernien2017probing,huang2021provably,bilkis2021semi,nghiem2020unified}, and hence no embedding scheme is needed (here, the main ingredient of the QML is the QNN). Moreover,  it is believed that QML models on quantum data can indeed achieve a quantum advantage~\cite{huang2021provably,kubler2021inductive}. Yet, despite the tremendous historical relevance of datasets in developing classical ML models, there is a dearth of truly quantum datasets for QML. Recently, the need for such datasets has been recognized, and small-scale quantum datasets on one and two qubits are now available~\cite{perrier2021qdataset}. However, more work still needs to be done to develop large-scale quantum datasets.

In this work we introduce a quantum dataset for QML composed of quantum states with different amounts and types of multipartite entanglement. We refer to it as the NTangled dataset. First, as shown in Fig.~\ref{fig:overview}(a), we train a QNN to generate states in the NTangled dataset. We test different ansatzes for the generator QNN, and we show how some properties of the ansatz can affect their performance. We then benchmark QML models for the supervised learning task of classifying the states in the NTangled dataset according to their level of entanglement (see Fig.~\ref{fig:overview}(b)). We additionally show that the same classifying QNNs can be used to classify states generated with hardware-efficient quantum circuits with different depths (the underlying mechanism here being that entanglement properties vary with circuit depth). Hence, states generated from hardware-efficient quantum circuits with different depths can also be considered as a scalable, easy-to-produce dataset.

At the basis of our dataset lies the fact that entanglement is the defining property of quantum mechanics, and is a fundamental resource for quantum computation~\cite{horodecki2009quantum,gigena2020one,sharma2020reformulation}. However, characterizing the multipartite entanglement in quantum states is in itself a difficult task~\cite{horodecki2009quantum} that can  benefit from the computational power of QML. In addition, generating varied multipartite entangled states is a difficult task for non-trivial system sizes. Here we introduce a novel method to do just that via QNNs. Hence, we also argue that the methods used to generate the NTangled dataset have relevance for quantum entanglement theory.

We finally note that in the Appendix and in the Github repository of Ref.~\cite{schatzki2021github}, we present a description and parameters for the ansatzes that will allow the reader to reproduce the NTangled dataset.   The datasets can be efficiently stored as they come in the form of a parametrized quantum circuit plus the values of the trained parameters.

\section{General Framework}~\label{sec:framework}
In this section we first introduce the general framework for benchmarking a QML model by solving  a paradigmatic supervised machine learning task: classification of labeled data. Here, we will further argue that one should benchmark QML models with quantum datasets rather than classical ones. Finally, we also provide a brief review of quantum datasets that have been used in the literature.

\subsection{Machine Learning}

In a supervised QML classification task, one is given a dataset of the form $\{\psi_i, y_i\}$. Here, $\psi_i=\dya{\psi_i}$, where $\ket{\psi_i}\in S\subseteq \HC$ are pure $n$-qubit quantum states in a set $S$ belonging to the  $d$-dimension Hilbert space $\HC$ (with $d=2^n$), and  $y_i\in Y$ are labels associated with each state $\psi_i$ according to some unknown model $h:S\rightarrow Y$. For the sake of simplicity, in this work we consider binary classification, so that $Y=\{0,1\}$. 

The first step is to create a training set $\TC=\{\psi_i, y_i\}_{i=1}^N$,   obtained by sampling states from the dataset according to a given probability distribution. Then, one trains a QML model, i.e. a parametrized function $f(\psi_i,\thv):S\rightarrow Y$, with the goal being that it accurately predicts the label of each quantum state. Ideally, one would want not only that the labels predicted by $f(\psi_i,\thv)$ match those predicted by $h$ in  $\TC$ (i.e. small training error), but also that $f(\psi_i,\thv)$ is capable of  making good predictions over new and previously unseen data (i.e., small generalization error).

The specific form of $f(\psi_i,\thv)$ depends on the task at hand, but it will generally  consist of the following steps. The model takes as input $m$ copies of the input state $\psi_i$, and possible $k$ ancillary qubits that are initialized to some  pure state $\phi=\dya{\phi}$. These states are then sent through a QNN $U(\thv_1)$, where $\thv_1$ can be composed of continuous parameters (such as trainable gate rotation angles), as well as discrete parameters (such as gate placements). At the end of the QNN, one measures  $r$ qubits, with $r\leq m\cdot n +k$, in the computational basis.  That is, one estimates the probabilities of obtaining measurement  outcomes for the bitstrings $\vec{z}\in\{0,1\}^{\otimes r}$ as
\begin{equation}
    p(\vec{z}|\psi_i,\thv_1)=\langle \vec{z}| U(\thv_1)\left(\psi_i^{\otimes m}\otimes \phi\right)U\ad(\thv_1)|\vec{z} \rangle\,.
    \label{eq:PQC}
\end{equation}
Note that here one can always apply a fixed unitary prior to the measurement to change the measurement basis. However, for simplicity we assume that one measures in the computational basis.

The measurement statistics are then post-processed via some classical function yielding predicted labels $\hat{y}\in[0,1]$. In this paper we consider linear post-processing achieved by a single layer neural network with a scalar output and sigmoid activation. The assigned  label $\hat{y}$ is then obtained from
\begin{equation}\label{eq:probability}
 \hat{y}=\sig\left(\sum_{\vec{z}}w_{\vec{z}}p(\vec{z}|\psi_i,\thv_1)+b\right)\,,   
\end{equation}
where $b$ (bias) and the coefficients $w_{\vec{z}}$ are real parameters, and where $\sig(\cdot)$ denotes the sigmoid activation function. Thus, we here define the set of trainable parameters as $\thv=\{\thv_1,\{w_{\vec{z}}\},b\}$. Finally, we note that one can also apply regularization to the classical post-processing to prevent overfitting. The latter can be achieved by adding a term into the loss function of the form $\lambda\norm{\vec{w}}$, where we took $\norm{\cdot}$ to be the $L1$ norm (i.e., lasso regularization).

To  quantify the success of the classification task over the training set $\TC$, one defines a loss function.  For instance, the mean-squared error loss function is of the form
\begin{equation}
    L(\thv)=\frac{1}{N}\sum_{i=1}^N(\hat{y}-y)^2\,.
\end{equation}
Then, the parameters in the QML model are trained by solving the optimization task
\begin{equation}\label{eq:optimization}
    \argmin_{\thv}L(\thv)\,.
\end{equation}
To test the accuracy of the model on the training and test sets, we further apply a Boolean function that yields  label $1(0)$ if $   \hat{y}\geq (<) 0.5$.

Let us here remark that when dealing with classical data, there is an extra step needed to encode the classical information in a quantum state. In this case,  the dataset is of the form $\{\vec{x_i},y_i\}$, where $\vec{x}_i\in X$ are  bitstrings of a given length, and one employs an embedding scheme $W(\vec{x}):X\rightarrow S $ that maps  bitstrings to quantum states~\cite{havlivcek2019supervised,lloyd2020quantum,schuld2020circuit,larose2020robust,abbas2020power,huang2021power}. Here, $n$  qubits are initialized to the fiduciary state  $\chi=\dya{\vec{0}}$ (where $\ket{\vec{0}}=\ket{0}^{\otimes n}$)  and  $W(\vec{x})$ is a unitary such that  $W(\vec{x}_i)\chi W\ad(\vec{x}_i)=\chi_{i}$. This embedding can be realized, for instance, with a parametrized quantum circuit whose rotation angles depend on the bits in $\vec{x}$. The latter  maps the original classical dataset to a set of the form $\{\chi_i,y_i\}$, after which the QML procedure previously presented can be employed. 

We note that when dealing with classical data, the embedding scheme can also be considered as part of the QML model.  Thus, the success in solving the optimization task of Eq.~\eqref{eq:optimization} does not only depend on the choice of QNN but also on the choice of embedding~\cite{schuld2021effect}.

\subsection{Benchmarking QML models}

\begin{figure}[t]
\includegraphics[width=1\linewidth]{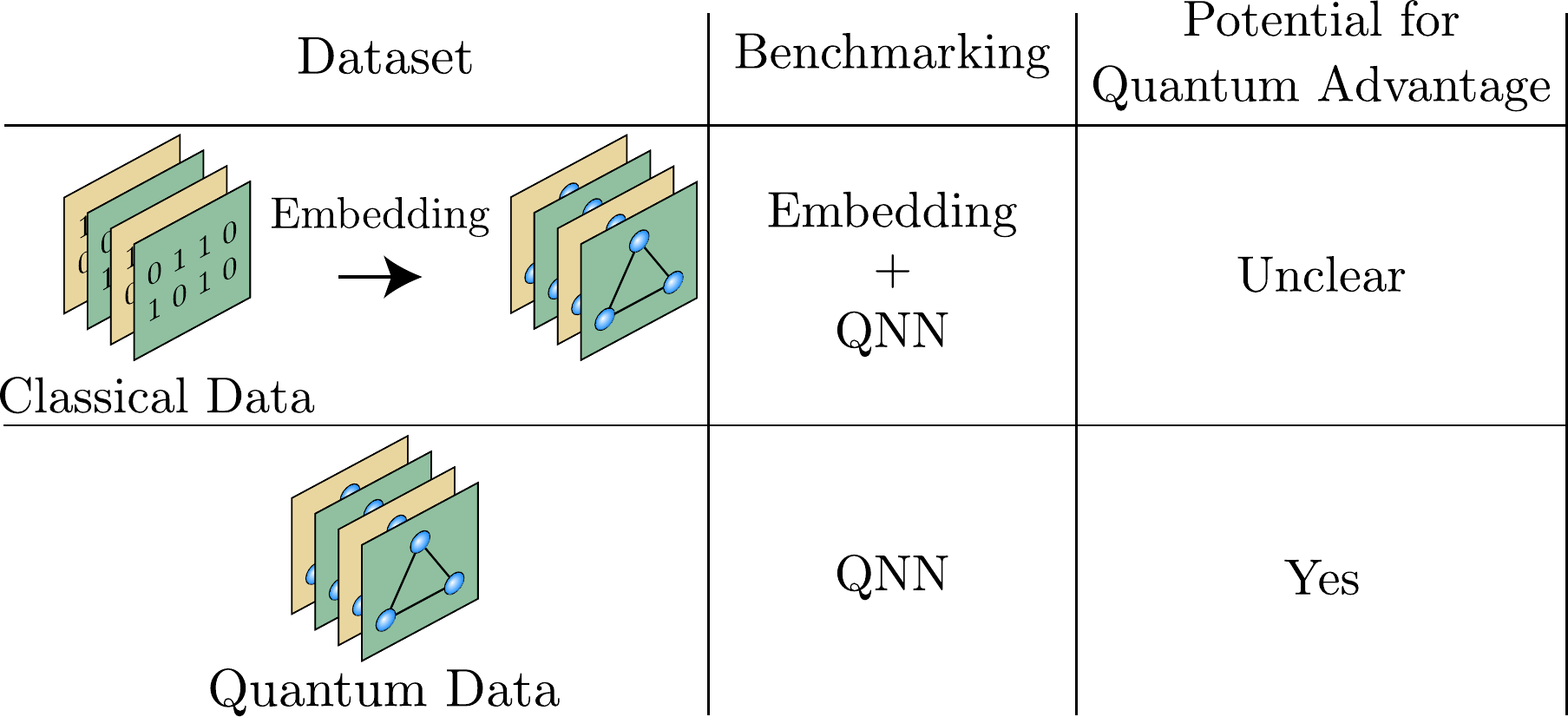}
\caption{\textbf{Classical versus quantum datasets.} Consider the task of benchmarking a QNN architecture by solving Eq.~\eqref{eq:optimization}. When using classical data, one is required to encode the data in a quantum state. Thus, a benchmarking experiment will simultaneously test the performance of the QNN and the embedding. This can lead to issues, as the choice of embedding can hinder one's ability to solve Eq.~\eqref{eq:optimization}. Moreover, it is still unclear whether a quantum advantage can be achieved with classical data. On the other hand, using quantum data, one only benchmarks the performance of the QNN by itself, and it is believed that a quantum advantage can indeed be reached using quantum data.}
\label{fig:2}
\end{figure}

Benchmarking machine learning models is a fundamental step to guaranteeing that  one can accurately optimize the model's parameters to achieve small training and generalization errors. In QML, the importance of efficient, large-scale benchmarking experiments is only amplified by the fact that QNNs can exhibit  trainability issues when scaled to large problem sizes.  For example, certain QNNs can exhibit the so-called barren plateau phenomenon~\cite{mcclean2018barren,cerezo2020cost,marrero2020entanglement,patti2020entanglement,larocca2021diagnosing,holmes2021connecting,wang2020noise,cerezo2020impact,holmes2021barren,arrasmith2020effect,arrasmith2021equivalence,wang2021can}, where the loss function gradients vanish exponentially with the system size. Since barren plateaus do not affect the trainability of small problem instances, it has been recognized that one needs to perform large scale implementations to benchmark the scaling of QNNs~\cite{sharma2020trainability}.

When benchmarking a QML model, one can test its parts separately and identify which aspect of the model can hinder its trainability. For instance, certain QNNs are more prone to barren plateaus than others~\cite{sharma2020trainability,pesah2020absence}, and measuring a large number of qubits can also be detrimental~\cite{cerezo2020cost}.

As shown in Fig.~\ref{fig:2}, with classical dataset one has  the added issue  that  the choice of embedding can also affect the trainability of the QNN parameters~\cite{thanaslip2021subtleties}. Moreover, since there is still an open debate regarding what the best encoding scheme $W(\vec{x})$ is~\cite{havlivcek2019supervised,perez2020data,lloyd2020quantum,abbas2020power,schuld2021effect,huang2021power}, it can be beneficial to benchmark QNNs and embeddings separately (e.g., first benchmark the QNN, and later benchmark the embedding).  In fact, many widely used embedding schemes are based on the fact that their outputs could be hard to classically simulate~\cite{havlivcek2019supervised}, but it is  still unclear if this criterion leads to practical and trainable  embeddings~\cite{thanaslip2021subtleties}. This problem can be avoided when using quantum datasets. Since there is no embedding, one can test the QNN by itself and check if it can be heuristically scaled to large problem sizes without incurring trainability issues. 

Finally, as previously mentioned, it is worth noting that it is still unclear whether QML can bring a quantum advantage when dealing with classical data. Recent results on quantum kernel-based QML models point towards the fact that a quantum advantage can be achieved using quantum data if one has access to the data generating process and if this process cannot be efficiently classically encoded~\cite{kubler2021inductive}. On the other hand, it is unclear whether such results can be applied to classical datasets~\cite{kubler2021inductive}. In addition, it has  been proved that, with some sampling assumptions, there exist classical algorithms performing as well (up to polynomial factors) as their quantum counterparts when dealing with classical data \cite{Tang2021}.

\subsection{Quantum datasets}

Let us here recall some quantum datasets that have been used in the literature to benchmark QML models. Perhaps the most natural way to produce quantum datasets is using the Variational Quantum Eigensolver (VQE) algorithm~\cite{peruzzo2014variational}, and forming the datasets from ground states of a given Hamiltonian. For instance, consider a Hamiltonian $H(p)$, where $p$ is a  tunable parameter such that there exists a critical value $p_c$ for which the ground states of $H(p)$ undergo a phase transition (the case when the Hamiltonian depends on multiple parameters  follows trivially). Then, one employs the VQE algorithm to find a parametrized unitary $V(\vec{\gamma})$ to solve the optimization problem
\begin{equation}
\vec{\gamma}_p=\argmin_{\vec{\gamma}} \bramatket{\vec{0}}{V(\vec{\gamma})H(p)V\ad(\vec{\gamma})}{\vec{0}}\,,   
\end{equation}
for different values of $p$. Here, each state in the dataset is prepared with a different unitary $V(\vec{\gamma}_p)$, and the state labels are $y_i=0$ ($y_i=1$) for all states with $p< p_c$  ($p> p_c$). Such VQE-based datasets have been formed with ground states of the Ising model~\cite{uvarov2020machine,banchi2021generalization}, the $XXZ$ Heisenberg spin chain~\cite{uvarov2020machine,huang2021provably}, Haldane chains~\cite{cong2019quantum}, Rydberg atom Hamiltonians~\cite{bernien2017probing,huang2021provably}, and the small molecules~\cite{bilkis2021semi,nakagawa2019github}.

A different approach to preparing quantum datasets is sampling output states from a given unitary. For instance, for binary classification one employs two unitaries, $V_0$ and $V_1$,  and prepares the quantum states in the dataset by applying each unitary to states in the computational basis. That is, all the states with label $y_i=0$  ($y_i=1$) are of the form $V_0\ket{\vec{z}_i}$ ($V_1\ket{\vec{z}_i}$), where $\ket{\vec{z}_1}$ are $n$-qubit states in the computational basis. This approach can be used to generate separable and mixed states~\cite{huang2021information}, and states generated as the output of hardware efficient ansatzes of different depths~\cite{grant2018hierarchical} (see Fig.~\ref{fig:depth_learning_ansatz} for a description of the hardware efficient ansatz). 

Finally, let us remark that recently, a collection of publicly available  quantum datasets was introduced in~\cite{perrier2021qdataset}. Therein, the authors create $52$  
datasets derived from simulating of one- and two-qubit systems evolving in the presence
and absence of noise. The work in Ref.~\cite{perrier2021qdataset} represents one of the first efforts to introduce quantum datasets. Nevertheless, the small-scale nature of these datasets still leaves open the need for scalable datasets that can benchmark QML models for large-scale problem sizes.

\section{NTangled quantum state dataset}

In this section we present the basis for a novel type of quantum dataset, called NTangled, that can be used to benchmark QML architectures. Our NTangled dataset is composed of quantum states with different amounts and different types of multipartite entanglement. 

\subsection{Multipartite Entanglement}

Entanglement has been recognized as one of the hallmarks of quantum systems~\cite{Einstein1935}. Moreover, it has also been shown that entanglement is a fundamental resource for quantum information processing tasks~\cite{horodecki2009quantum,gigena2020one,sharma2020reformulation}. Non-entangled product states are those that can be expressed as a tensor product of single-qubit states, i.e., $\ket{\psi}$ is a product state on $n$ qubits iff $\ket{\psi}=\bigotimes_{i=1}^n\ket{\psi_i}$. A pure state is said to be entangled if it cannot be written in this form. For bipartite systems, entanglement is well-understood and can be completely characterized via the Schmidt decomposition~\cite{ekert1998entangled,horodecki2009quantum,walter2016multipartite}. On the other hand, quantifying and characterizing multipartite entanglement proves far more complicated. For instance, already for a system of $3$ qubits it has been shown that there exists two different, and inequivalent, types
of genuine tripartite entanglement~\cite{dur2000three}.

While there are multiple ways to quantify the presence of multipartite entanglement in quantum states, we consider here the Concentratable Entanglement (CE) of Ref.~\cite{beckey2021computable}.  For a pure state $\ket{\psi}$, we here employ the following definition of the CE 
\begin{equation}
    \mathcal{C}(\ket{\psi})= 1-\frac{1}{2^n}\sum_{\alpha\in Q}\Tr[\rho_\alpha^2],
    \label{eq:CE}
\end{equation}
where $Q$ is the power set of the set $\{1,2,\ldots,n\}$, and where $\rho_\alpha$ is the reduced state of $\ket{\psi}$ in the subsystems labeled by the elements in $\alpha$ (here $\rho_\emptyset:=1$). As shown in~\cite{beckey2021computable}, unlike other multipartite entanglement measures, the CE  can be efficiently computed on a quantum computer via a generalized controlled SWAP test~\cite{beckey2021computable}. We refer the reader to the appendix for further details on how to compute the CE and for other properties of multipartite entanglement. Therein we further show how the $n$-tangle $\tau_n$~\cite{Wong2001} (another entanglement measure) can also be estimated, simultaneously, from the same circuit that estimates the CE.

As proved in the Appendix, we derive the following continuity bound for CE:
\begin{theorem}\label{theo:1}
    Given two $n$-qubit pure states $\psi=\dya{\psi}, \phi=\dya{\phi}$, the difference between their CE follows
    \begin{align}
        |C(\ket{\psi})-C(\ket{\phi})| &\leq \sqrt{2}(2-2^{2-n})D_{tr}(\psi,\phi)\\
        & \leq 2\sqrt{2}D_{tr}(\psi,\phi)\,,
        \label{eq:CE_Theorem}
    \end{align}
where, for matrices $A$ and $B$, $D_{tr}(A,B)=\frac{1}{2}\Tr[\sqrt{(A-B)\ad (A-B)}]$ denotes the trace distance.
\end{theorem}

Below we discuss how Theorem~\ref{theo:1} has direct implications for training a generative model that outputs states with similar CE values.

\subsection{Generation of the NTangled dataset}

Here we consider the task of generating an NTangled dataset by training a QNN model to output states with a desired value $\xi$ of CE. Specifically, given a parametrized quantum circuit $V_G(\vec{\alpha})$ (here the sub-index $G$ denotes the fact that the QNN generates the dataset), and given a probability distribution of quantum  states $\PC$, the goal is to train the parameters $\vec{\alpha}$ such that
\begin{equation}
    \mathop{\mathbb{E}}_{(\ket{\psi^{(\tin)}})\sim\PC}\left[C\left(\ket{\psi^{(\tout)}}\right)\right]=\xi\,,
\end{equation}
where we have defined $\ket{\psi^{(\tout)}}=V_G(\vec{\alpha})\ket{\psi^{(\tin)}}$.

The parameters $\vec{\alpha}$ are trained by minimizing the mean squared loss function over a training set $\TC$ of $N$ input states sampled from $\PC$:
\begin{equation}\label{eq:loss-generation}
    L(\vec{\alpha})=\frac{1}{N}\sum_{\ket{\psi_i^{(\tin)}}\in\TC} \left(C\left(\ket{\psi_i^{(\tout)}}\right)-\xi\right)^2\,.
\end{equation}
Finally, once the parameters are trained, we evaluate the generalization error of the QNN  by considering as a success the cases when the output state has a CE value in $[\xi-\delta,\xi +\delta]$, where $\delta$ is a fixed allowed error value. 

We here note that, as previously mentioned, there are different non-equivalent types of multipartite entangled states when $n> 2$. Moreover, while the CE can quantify the amount of entanglement in a state, it does not characterize the type of entanglement. For instance, both $W$-type states and $GHZ$-type states can have the same amount of CE despite their entanglement properties being different. Thus, if one wished to bias the output probability distribution so that a type of state is more probable than another type, then one can add a term in the loss function that penalizes undesirable types. As a specific example, one can reward or penalize the presence of $W$-type states by modifying Eq.~\eqref{eq:loss-generation} as
\begin{align}\label{eq:loss-generation-tangle}
\begin{split}
    L(\vec{\alpha},c_1,c_2)=\frac{1}{N}\sum_{\ket{\psi_i^{(\tin)}}\in\TC}\Big[&c_1 \left(C\left(\ket{\psi_i^{(\tout)}}\right)-\xi\right)^2\\
    &+c_2\tau_n\left(\ket{\psi_i^{(\tout)}}\right)\Big]\,,
\end{split}
\end{align}
where $c_1$ and $c_2$ are real numbers, and where $\tau_n\left(\ket{\psi_i^{(\tout)}}\right)$ is the $n$-tangle of the output state. Here we recall that $\tau_n\left(\ket{W}\right)=0$ for a $W$-state \cite{wong2001potential}.

A crucial aspect for generating the NTangled dataset is the choice of probability distribution $\PC$ from which one samples the input states for training the QNN  $V_G(\vec{\alpha})$. First, let us note that it is not obvious that one can train a QNN to produce states with arbitrary levels on entanglement by letting $\PC$ be an arbitrary (or random) set of states. Hence, we here choose $\PC$ to be composed of states that are equivalent under local operations~\cite{horodecki2009quantum}. This will guarantee that these states  have the same amount of initial multipartite entanglement and that $\PC$ has measure zero in the Hilbert space. 

Let us now describe the approaches we used for choosing $\PC$. A first choice is having $\PC$ be the set of computational basis states. While this choice is natural, it leads to the issue that while the size of this set is $2^n$, for small problem sizes one is not able to generate a sufficiently large number of states (e.g. for $n=3$ qubits, the set only contains $8$ states). Hence, an alternative here is setting $\PC$ to be the set of product states. The main advantage here is that the set of product states is continuous and thus of infinite cardinality for all problem sizes. 

Let us finally remark that a third alternative is to train the generative QNN on computational basis states, but then test its generalization on more general product states. Note that this approach is different from the usual machine learning scheme where the training and testing set are sampled from the same probability distribution. 
However, as shown in the numerical results section, this approach leads to low generalization errors, and thus also constitutes a viable option.

\begin{figure}[t]
    \centering
    \includegraphics[width=1\columnwidth]{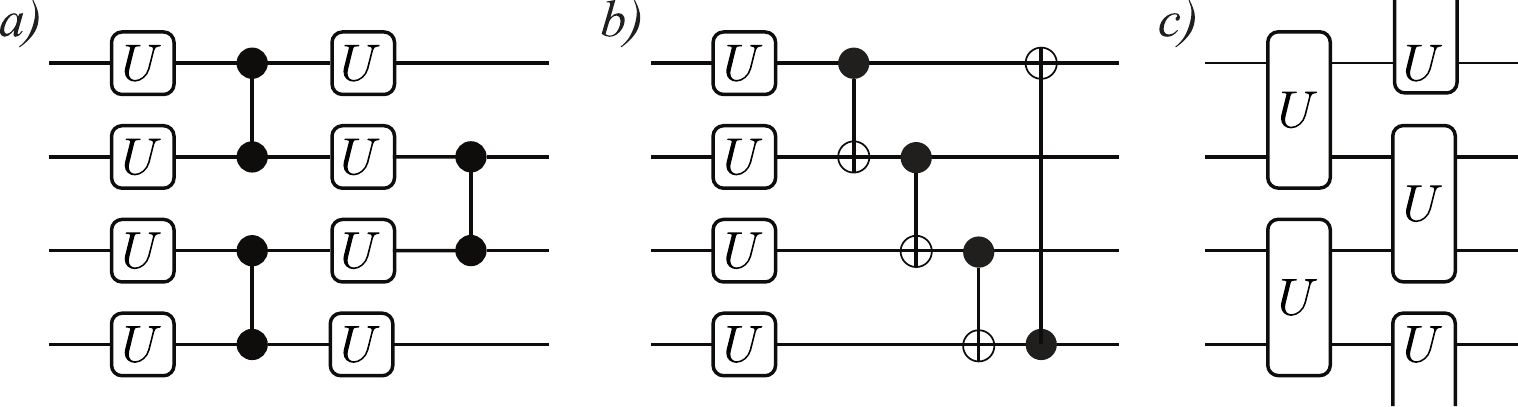}
    \caption{\textbf{Ansatzes for the generator QNN  $V_G(\vec{\alpha})$.} Shown are examples for generating entangled states on $n=4$ qubits. (a) A single layer of the HWE ansatz is composed of arbitrary single qubit unitaries followed by CZ gates acting on neighbouring pairs of qubits. (b) A single layer of the SEA ansatz contains arbitrary single qubit unitaries with a wrapping ladder of CNOTs. (c) A single layer of the CONV ansatz is composed of general two qubit gates acting on alternating pairs of qubits in a brick-like fashion. Note that this ansatz contains the HWE as a special case. }
    \label{fig:generator_ansatzes}
    
\end{figure}

We finally note that one can further improve the performance of the generative QNN using Theorem~\ref{theo:1} and letting $\PC$ sample states in an $\varepsilon$-Ball. Specifically, consider the following proposition.
\begin{proposition}\label{prop:1}
Let $\PC$ be a probability distribution of quantum states in an $\varepsilon$-Ball. That is, for any two pure quantum states $\psi=\dya{\psi}$ and $\phi=\dya{\phi}$ sampled from $\PC$, their trace distance is such that $D_{tr}(\psi,\phi)\leq \varepsilon$. Then, for any two states $\psi=\dya{\psi}$, $\phi=\dya{\phi}$ sampled from $\PC$ and sent through the  QNN $V_G(\vec{\alpha})$, their CE difference is upper bounded as 
\begin{equation}
    |C(V_G(\vec{\alpha})\ket{\psi})-C(V_G(\vec{\alpha})\ket{\phi})|\leq 2\sqrt{2}\varepsilon.
\end{equation}
\end{proposition}
\begin{proof}
    The proof of Proposition~\ref{prop:1} follows by noting that the trace distance is preserved by the action of the unitary QNN, i.e., $D_{tr}(\psi,\phi)=D_{tr}(V_G(\vec{\alpha}){\psi}V_G\ad(\vec{\alpha}),V_G(\vec{\alpha})\phi V_G\ad(\vec{\alpha}))$, and by applying Theorem~\ref{theo:1}.
\end{proof}
Proposition~\ref{prop:1} implies that the distribution of CE entanglements has a width that can be upper bounded by the width of the $\varepsilon$-Ball from which  $\PC$ samples. 

Here, we find it convenient to note that the NTangled dataset presented in the Github of~\cite{schatzki2021github} is composed of the trained parameters in $V_G(\vec{\alpha})$ plus a description of what the distribution $\PC$ is. In this sense, users can sample states $\ket{\psi^{(\tin)}}$ from $\PC$ and generate entangled states as  $\ket{\psi^{(\tout)}}=V_G(\vec{\alpha})\ket{\psi^{(\tin)}}$.

\begin{figure}[t]
    \includegraphics[width=1\linewidth]{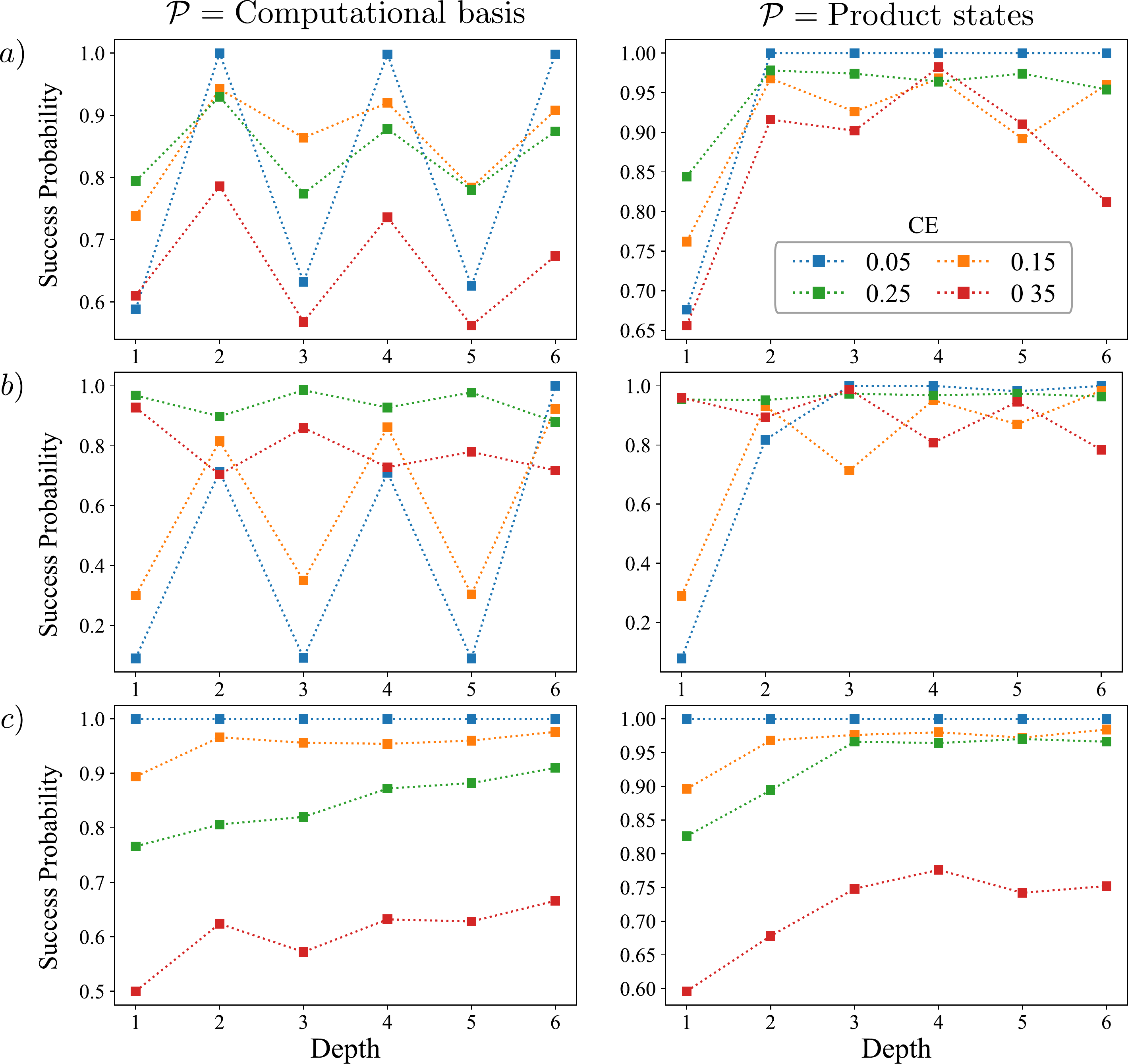}
    \caption{\textbf{Success probability for different generator QNN ansatzes with varying depths.} Here we show the success probability for generating $3$ qubit states in the NTangled dataset with entanglement $\xi\in \{0.05, 0.15, 0.25, 0.35\}$,   allowed error $\delta=0.1$, and when using the HWE (a), SEA (b), and CONV (c) ansatzes. The left column  corresponds to training with computational basis states, while the right column corresponds to training with product states. In both cases, we tested on product states. }
    \label{fig:3q_success_rates}
\end{figure}
\section{Numerical Simulations}
Here we present numerical results for training a QNN to generate the states in the NTangled dataset. We then use the NTangled dataset to train a classifying QML model. Finally, we present results for classifying states obtained from hardware efficient (HWE) quantum circuits (see Fig.~\ref{fig:generator_ansatzes}) with different depths. 

In the Appendix and in the Github repository of Ref.~\cite{schatzki2021github}, we present a description of the ansatz for the QNNs that generate the NTangled dataset, and we also include the value of trained parameters. Taken together, the results should allow the reader to use the NTangled dataset and reproduce our results. 

Here we note that for all numerical simulations in this section we employed the ADAM optimizer to train the QNN parameters~\cite{kingma2015adam}. Moreover,  the simulations were performed in the absence of hardware and shot noise through Tensorflow Quantum/Cirq~\cite{broughton2020tensorflow} and Pennylane~\cite{bergholm2018pennylane}. Further analysis of entanglement properties were performed using QuTIP~\cite{johansson2012nation}.

\subsection{Entangled State Generation}

Let us first introduce the ansatzes for the generator QNN  $V_G(\vec{\alpha})$. As shown in  Fig.~\ref{fig:generator_ansatzes}, we consider three different architectures. For simplicity we refer to these ansatzes as (a) the hardware-efficient ansatz (HWE), (b) the strongly-entangling ansatz (SEA) \cite{Schuld2020}, and (c) the convolutional ansatz (CONV). In the figure, each box labeled with a $U$ represents an arbitrary one- or two-qubit gate. The one-qubit unitaries contain $3$ parameters, while the two-qubit ones are parametrized with $15$ rotation angles~\cite{vatan2004optimal}. We present in the Appendix the specific gate decomposition of these general unitaries. 

\subsubsection{Three qubit NTangled dataset}

Here we consider  the task of generating an NTangled dataset on $3$ qubits. For this purpose, we employ the loss function of  Eq.~\eqref{eq:loss-generation}. Then, we set the value of the allowed error $\delta=0.1$, so that the generator QNN was able to successful generate an output state with entanglement $\xi$  if $|\mathcal{C}(\ket{\psi_i^{(\tout)}})-\xi|\leq 0.1$.

\begin{figure*}[t!]
\centering
\includegraphics[width=1\linewidth]{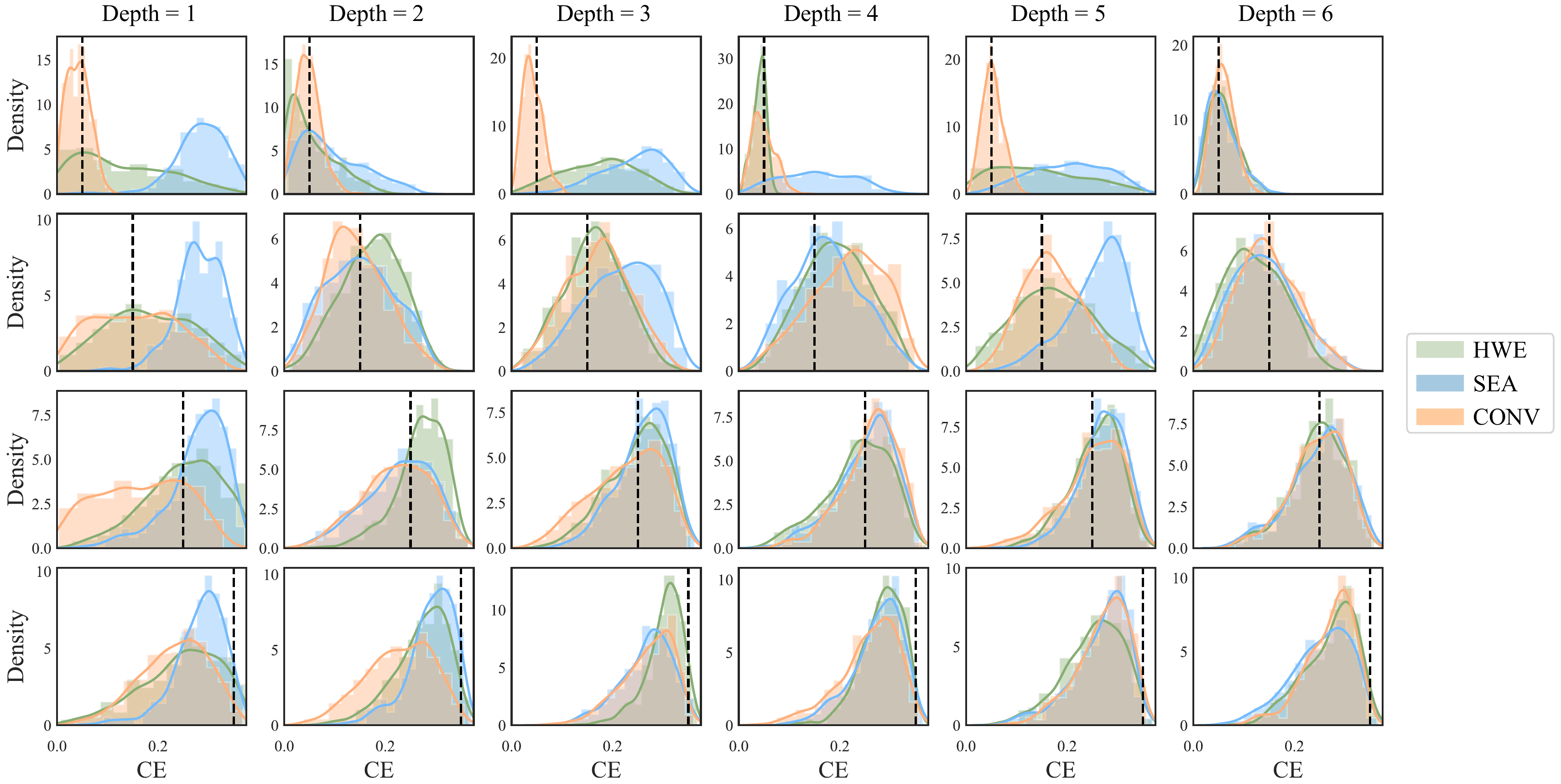}
    \caption{\textbf{Distribution of entanglement for the $3$ qubit NTangled dataset.} We show the probability density function of the output state's CE for the same models  of Fig.~\ref{fig:3q_success_rates} when trained on product states. Each color repents a different ansatz. Plots in the same column correspond to ansatz with the same depth (indicated above), while plots in the same row correspond to training with a given target CE calue, $\xi \in \{0.05, 0.15, 0.25, 0.35\}$ (indicated with a vertical dashed line). We refer the reader to the Appendix for a similar plot showing results obtained when training with input computational basis states.}
    \label{fig:3q_hwe_ps_entanglements}
\end{figure*}

In Fig.~\ref{fig:3q_success_rates} we present results for training the QNNs in Fig.~\ref{fig:generator_ansatzes} with varying depths (i.e., number of layers) to generate datasets composed of states with CEs $\xi \in \{0.05, 0.15, 0.25, 0.35\}$.  Then, as proved in the Appendix, the following result  holds.
\begin{theorem}\label{thm:3q_ce}
    Any 3 qubit state $\ket{\psi}$ with $\mathcal{C}(\ket{\psi}) > 0.25$ has true multipartite entanglement.
\end{theorem}
Theorem~\ref{thm:3q_ce} allows us to guarantee that some of the  states in the NTangled datasets generated  for $\xi=0.25,0.35$ have true multipartite entanglement. That is, no qubit in the state is in a product state with the others.

In the left column of Fig.~\ref{fig:3q_success_rates} we show results for training with input states sampled from the computational basis, while on the right column we trained with input states randomly sampled from the set of product states. For computational basis states, we trained on $n+1=4$ states (one hot encoded states plus the all zero state), while for product states we trained in $n^2+1=10$ states.  In all cases, the success probabilities were tested on  $500$ random product states. Moreover, for each tuple of ansatz, depth, and goal CE, we trained $50$ models and selected the one achieved the smallest training error. We note that the case when we trained and tested on computational basis input states (not shown in Fig.~\ref{fig:3q_success_rates}), we can achieve $100\%$ generation success (zero training and generalization error).

As seen in  Fig.~\ref{fig:3q_success_rates}(a) and (b), the HWE and SEA have an oscillating behaviour in the number of layers, with an even number of layers outperforming odd number of layers. Moreover, as expected, we find that for all ansatzes, training and testing on the same set leads to higher success probability (right column). However, we note that one can achieve remarkable performances (with $>80\% $ success probability) by training on computational basis states, and then testing on product states (left column). Finally, here we can see that the SEA generally outperforms the HWE and CONV ansatzes for states with large amounts of entanglement. This shows that depending on the amount of entanglement one wishes to produce, it might be convenient to use one ansatz over the other.  Hence, for all targeted CE values we can choose an ansatz such that the success probability is always above $95\%$.

Here we note that Fig.~\ref{fig:3q_success_rates} only shows the probability that the CE of the generated state is within some error $\delta$ of the targeted entanglement value. To better understand the entanglement distribution of the generated states, we show in Fig.~\ref{fig:3q_hwe_ps_entanglements} the probability density function of the output state's CE for each tuple of ansatz, depth, and goal CE. Here we can clearly see that in most cases the densities are always peaked around the desired CE value (indicated by a dashed vertical line). For low entanglement states, the CONV ansatz performs the best (irrespective of the depth).  On the other other hand, as expected, we can see that for high entanglement, the SEA achieves highly peaked distributions already at one layer.  Finally, we remark  that for six layers the distributions of all ansatzes are similar. This is due to the fact that after a certain depth  the expressibility~\cite{holmes2021connecting} of  the parametrized quantum circuits is about the same.

\subsubsection{Four qubit NTangled dataset}

Here we consider the task of generating a $4$ qubit NTangled dataset. Similarly to the previous case, we train a generator QNN for datasets with CE $\xi\in\{0.05, 0.15, 0.25, 0.35\}$, and we again consider an allowed error of $\delta=0.1$. 

In Fig.~\ref{fig:4q_success_rates} we show results for training a HWE and SEA ansatz. The QNNs in the left column of Fig.~\ref{fig:4q_success_rates}  were trained with input states sampled from the computational basis, while on the right column they were trained with input states randomly sampled from the set of product states. For computational basis states, we again trained on $n+1=5$ states (one hot encoded states plus the all zero state), while for product states we trained in $n^2+1=17$ states.  In all cases, the success probabilities were tested on  $500$ random product states. Similarly to the $3$ qubit dataset, for each tuple of ansatz, depth, and goal CE, we trained $50$ models and selected the one achieved the smallest training error. 

\begin{figure}[t]
    \includegraphics[width=\columnwidth]{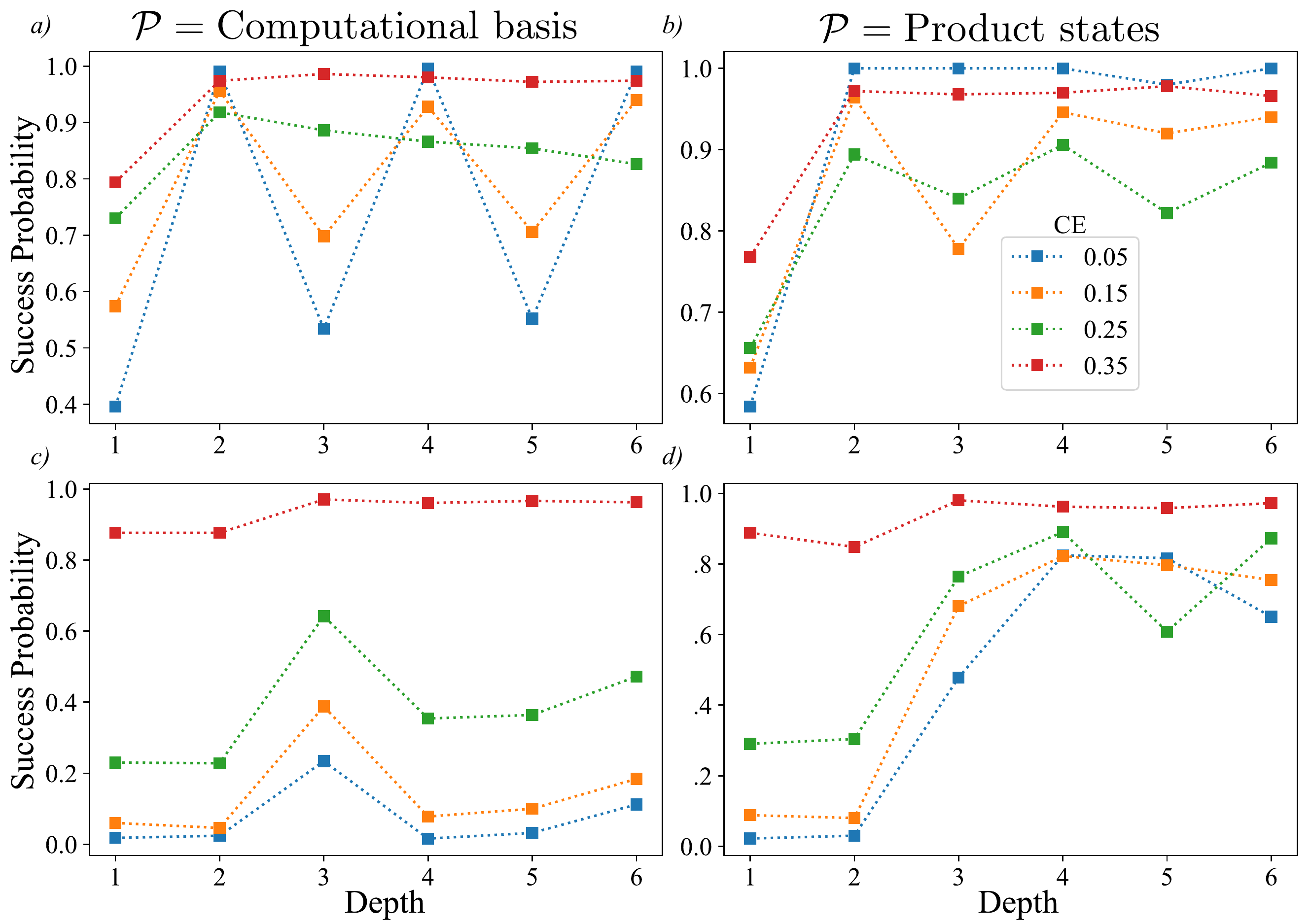}
    \caption{\textbf{Success probability for different generator QNN ansatzes with varying depths.} Here we show the success probability for generating $4$ qubit states in the NTangled dataset with entanglement $\xi\in \{0.05, 0.15, 0.25, 0.35\}$,  allowed error $\delta=0.1$, and when using the HWE (a), and SEA (b) ansatzes.  The left column  corresponds to training with computational basis states, while the right column to training with product states. }
    \label{fig:4q_success_rates}
\end{figure}

First, let us remark that when training and testing on computational basis states, we can obtain $100\%$ success probability. Moreover, we again see that the SEA ansatz strongly biases the distributions to highly entangled states.  On the other hand, the HWE with two or more layers is again able to produce states with different amounts of entanglement with high probability, specially when training and testing on random product states. Finally, by appropriately choosing the ansatz and depth, we can obtain states with a desired value of entanglement with probabilities larger than $91\%$. 

In the Appendix we show the entanglement distribution for the states obtained with the ansatzes of Fig.~\ref{fig:4q_success_rates}. There, we again can see that the distributions are always peaked around the desired value of CE.

\begin{figure}[t]
    \centering
    \includegraphics[width=1\columnwidth]{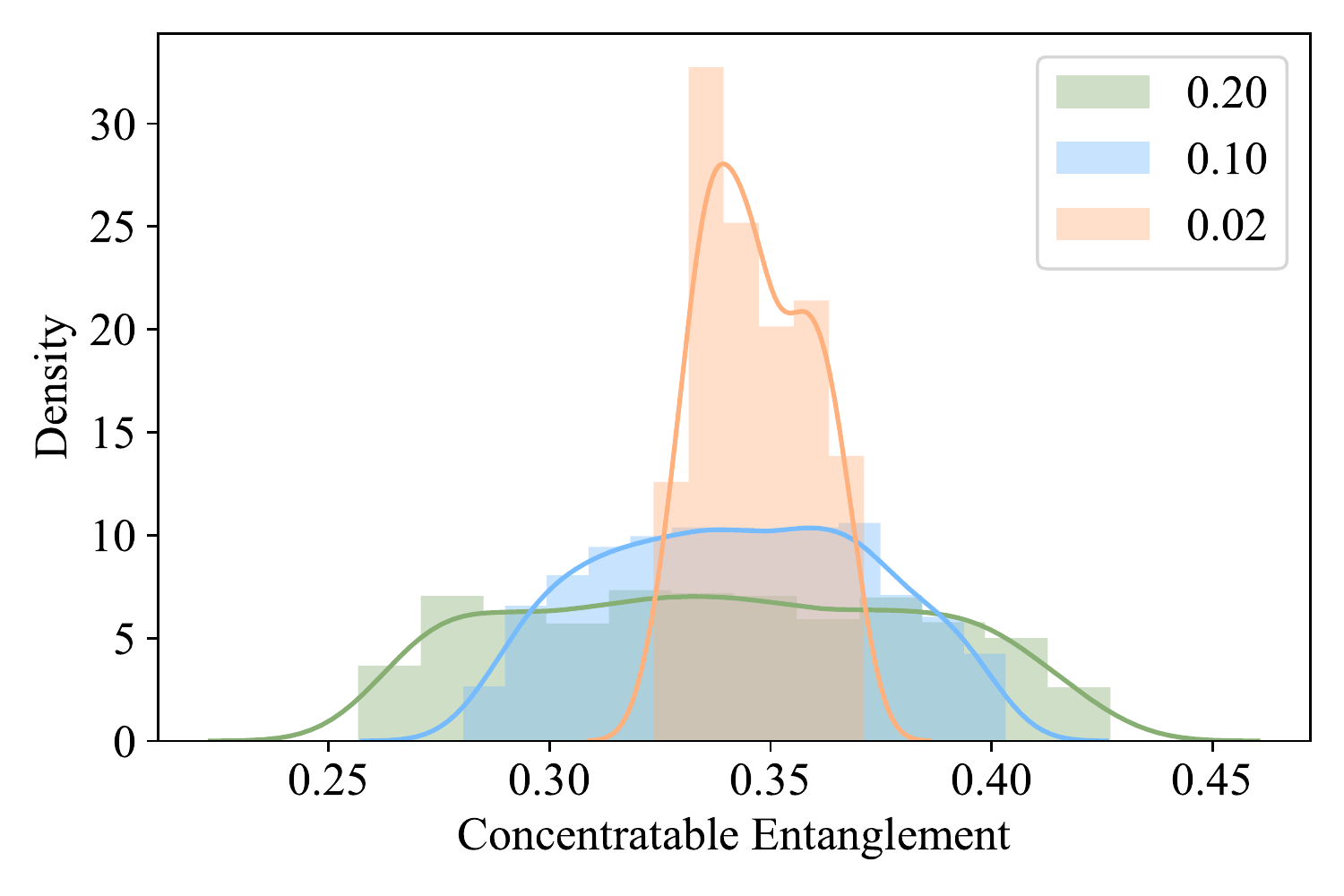}
    \caption{\textbf{Concentratable entanglement of states in a $\varepsilon$-Ball.} Here we trained a 4 qubit HWE QNN (on the computational basis) to output states with $\zeta = 0.35$. The inpu states to the QNN where obtained in a $\varepsilon$-Ball by applying small local rotations to a fixed input $\ket{\textbf{0}}$. We considered $\varepsilon\in\{0.2,0.1,0.02\}$. }
    \label{fig:eps-ball}
\end{figure}

Finally, in Fig.~\ref{fig:eps-ball} we present results which use Proposition~\ref{prop:1} to tighten the generated states CE distribution. Namely, we sent as input to a HWE QNN a set of states within an $\varepsilon$-Ball of the all-zero state $\ket{\vec{0}}$. Here we can see that, as expected, when $\varepsilon$ decreases the CE distribution tightens and becomes more peaked.

\subsubsection{Eight and twelve qubit NTangled dataset}

We now move on to larger dataset sizes:  $8$ and $12$ qubits. Unlike the $3$ and $4$ qubit NTangled datasets, here the Hilbert space is sufficiently large that one can train and test only on computational basis states. Moreover, we here again consider training on a polynomial number of input states. Surprisingly, as we see below, training on these small training sets still allow us to obtain good generalization results in the exponentially large Hilbert space.

\begin{figure}[t]
    \centering
    \includegraphics[width=\columnwidth]{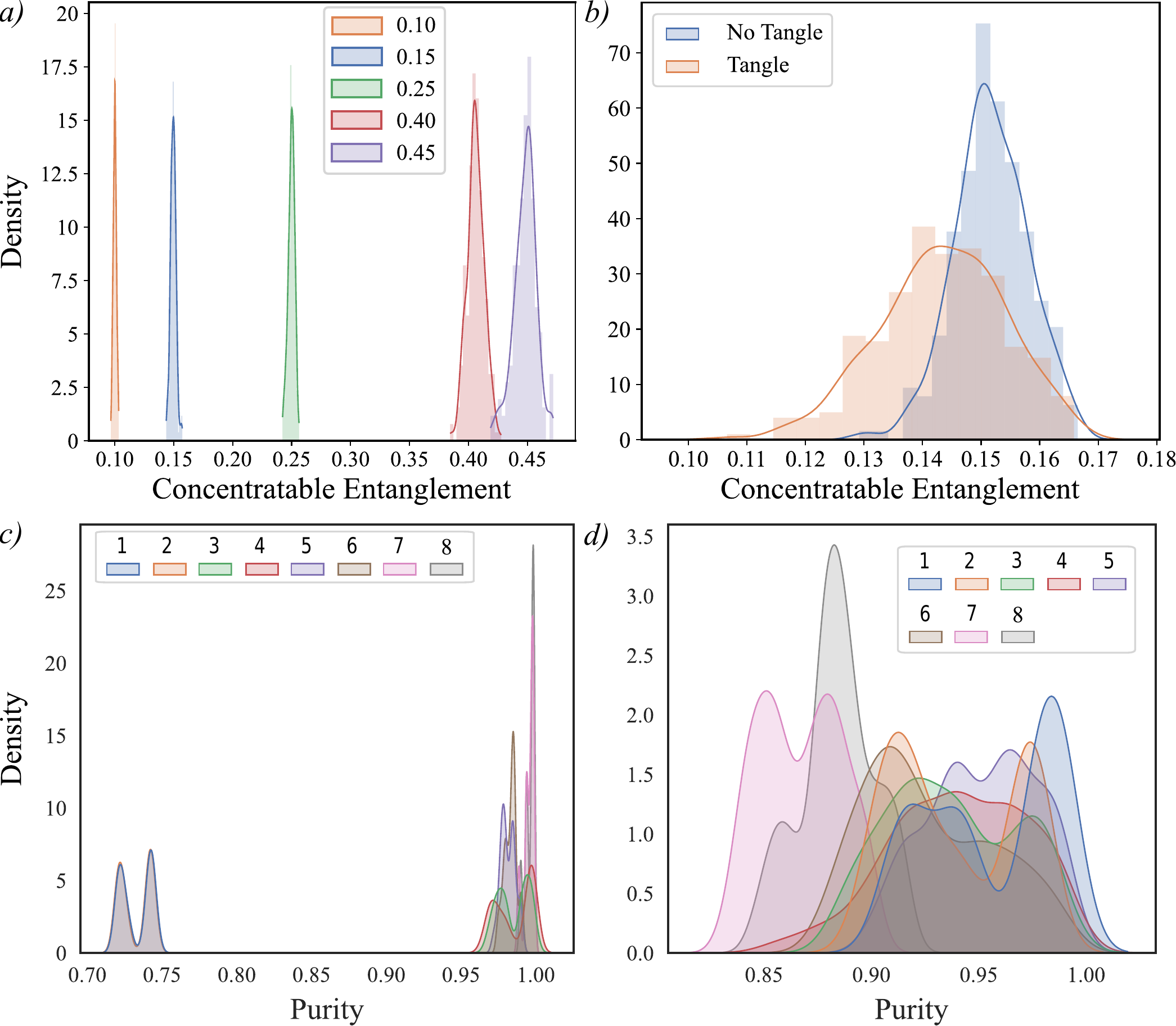}
    \caption{\textbf{Eight qubit NTangled dataset characterization.} a) Entanglement probability density function for the states in the $8$ qubit NTangled dataset with target CE values of $\xi \in \{0.1, 0.15, 0.25, 0.4, 0.45\} $.  b) Entanglement probability density function for the states with target CE value of $\xi = 0.15$ trained with and without an $n$-tangle term in the loss function according to Eq.~\eqref{eq:loss-generation-tangle}. Note that the blue curve represents the same data as that in panel a). c) Purity distribution of the single-qubit states generated in panel (b, blue curve), when training with no $n$-tangle in the loss function.   c) Purity distribution of the single-qubit states generated in panel (b, orange curve), when training with a $n$-tangle in the loss function. In panels b), and c), the curve labeled with index $i$ represents the purity distribution of the $i$-th qubit reduced states.}
    \label{fig:8q_graphs}
\end{figure}

For $8$ qubits we trained $10$ instances of HWE generator  QNNs with $5$ and $6$ layers on $n^2$ input states, and we picked the one that achieved the smallest training error. Here, when generalizing to the full computational basis we are able to obtain  $100\%$ success rate for all CE values of $\xi \in \{0.1, 0.15, 0.25, 0.4, 0.45\} $ with an allowed error of $\delta=0.1$. As we can see in Fig.~\ref{fig:8q_graphs}(a), the distributions of CE are well peaked around the targeted values.

Here it is convenient to take a closer look at how the entanglement is distributed in the states generated by the QNN, specially when we target a small CE value. In this case, it could happen that the entanglement is not properly distributed among all possible parties, but rather is concentrated on a subset of qubits, with the rest of the qubits being in a quasi-product state.

To test how the entanglement is distributed, we compute the  purities of the single-qubit reduced states  for the  dataset generated with a target value of $\xi=0.15$. As shown in Fig.~\ref{fig:8q_graphs}(c), most of the single-qubit reduced states  have a purity that is very close to one, indicating that in general they are almost pure, i.e. all but a small number of parties are unentangled. Thus, despite hitting the correct CE value, states of this type are not necessarily interesting from the perspective of multipartite entanglement.

To bias the entanglement distribution so that the entanglement is more distributed across all parties, we trained the QML model by  adding the  $n$-tangle into the loss function according to Eq.~\eqref{eq:loss-generation-tangle}. In Fig.~\ref{fig:8q_graphs}(b) we show the entanglement probability density function obtained by training with and without the $n$-tangle in the loss function. Here we can see that by adding the $n$-tangle, the distribution widens, but all states are still within the allowed error of $\delta=0.1$. Then, in Fig.~\ref{fig:8q_graphs}(d) we show the purities of the single-qubit reduced states. We can now see that the previous quasi-pure state behaviour no longer appears as all qubits have a wide range of purity, meaning that the entanglement is well distributed across the whole system.

Finally, for a $12$ qubit implementation, we trained a HWE generator QNN with $10$ layers on a training set composed of only $n+1=13$ states. To check the optimizer performance in a large scale experiment, we  ran a single instance of the QNN parameters optimization. Here, we targeted a CE value of $\xi=0.25$, and with an allowed error of $\delta=0.1$, we were able to  accomplish a $75\%$ success rate on the entire $4096$-dimensional computational basis.

\subsection{Entangled State Classification}

In this section we showcase using the states in NTangled datasets to benchmark a QML model for the supervised quantum machine learning task of classifying states according to their entanglement. Specifically, we use the $3$ and $4$ qubit NTangled datasets with target CE value of $\xi \in \{0.05, 0.25\} $, and for generator QNNs trained and tested on random product states. As we can see in Fig.~\ref{fig:3q_hwe_ps_entanglements} for $3$ qubits, the CE densities in those datasets  are very peaked and have low overlap probability (under $3\%$ overlap). This property can also be seen in the Appendix for the CE density distributions of the $4$ qubit dataset.  Then, for the classifying QNN, we employ the Quantum Convolutional Neural Network (QCNN) of Ref.~\cite{cong2019quantum}, as it has been shown that this architecture does not exhibit barren plateaus~\cite{pesah2020absence}. 

As stated in the General Framework (Section~\ref{sec:framework}), the classifying QML model can take as an input $m$-copies of the each state in the dataset. To test the performance of the classifying QML as a function of the number of copies $m$, we first consider the case when the QCNN only takes as input a single copy of each training state. Here, one can rarely achieve over $70\%$ classification accuracy. This result is expected, as accessing entanglement information in a pure quantum state usually requires the evaluation of a polynomial of order two in the matrix elements of the state's density matrix~\cite{walter2016multipartite,horodecki2009quantum,subacsi2019entanglement}. For instance, estimating in a quantum circuit the purity of a pure quantum state (which measures the amount of mixedness, and quantifies the bipartite entanglement), requires two copies of the input state~\cite{cincio2018learning,beckey2021computable}. Thus, it is natural to consider that simultaneously inputting to the QCNN two copies of each state in the dataset should improve its performance.

\begin{figure}[t]
    \centering
    \includegraphics[width=\columnwidth]{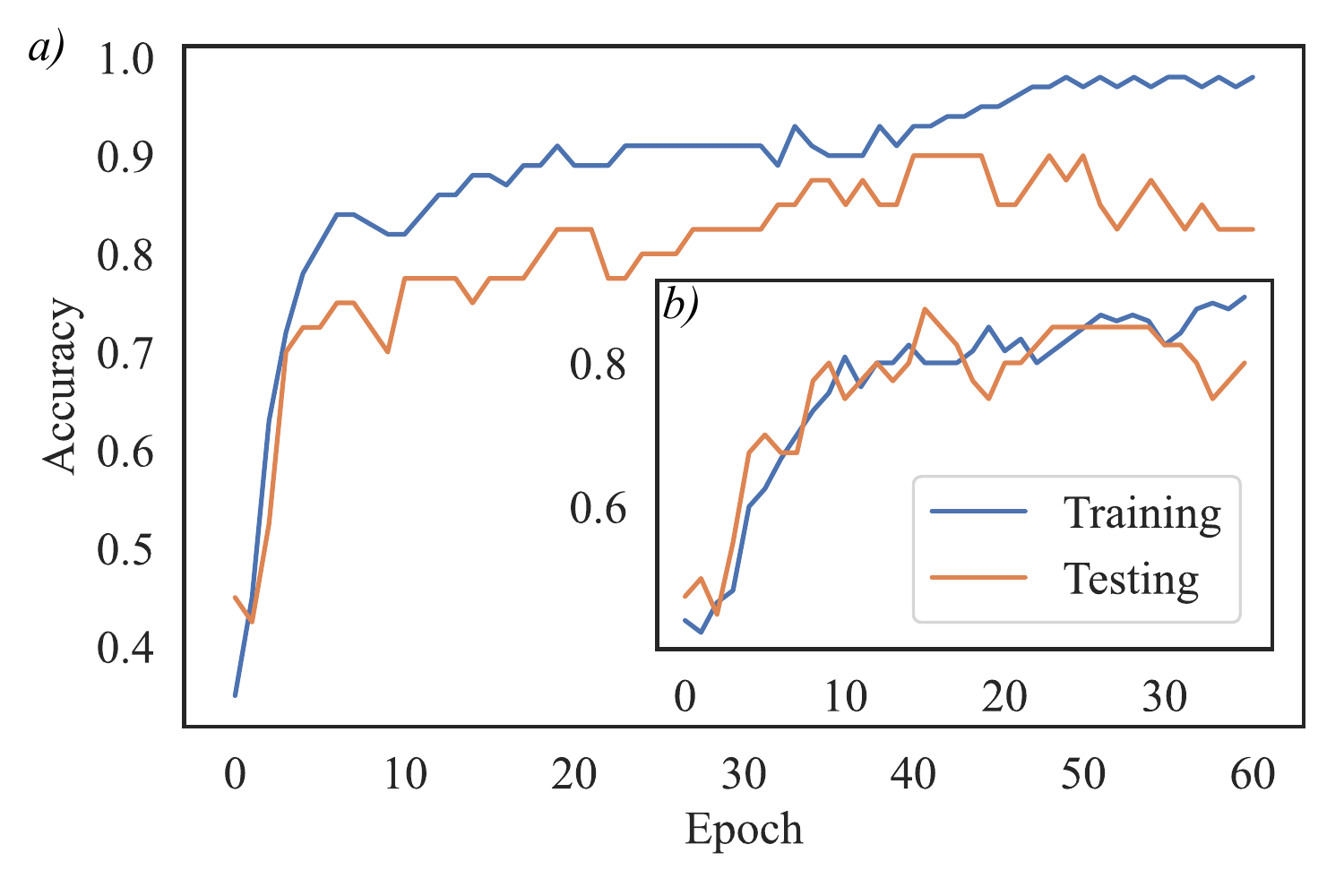}
    \caption{\textbf{Classification of states in the NTangled datasets.} We show the training and testing accuracy versus the number of training epochs for training on (a) $4$, and (b) $3$ qubit NTangled datasets. In this classification task, we train a QCNN to distinguish between states in the NTangled datasets with target CE values of $\xi \in \{0.05, 0.25\} $. The QCNN took as input two copies of each state in the dataset. }
    \label{fig:classification}
\end{figure}

In Fig.~\ref{fig:classification} we present results for training a  QCNN  by using $m=2$, i.e.,  feeding two copies of each input state. The results shown are obtained by independently training $15$ models and picking the one that reached the largest training accuracy. Here we see that for both $4$ qubits (Fig.~\ref{fig:classification}(a)) and $3$ qubits (Fig.~\ref{fig:classification}(b)) we can achieve a testing classification accuracy of over $90\%$. First, let us note the significant improvement that arises by using two copies rather than one. Second, we remark that  achieving $100\%$ accuracy is untenable due to the  $3\%$ mislabeling error.


\subsection{Ansatz Depth Learning}

\subsubsection{Motivation}

We here explore a different way of generating entangled datasets by sampling output states from quantum circuits with different depths~\cite{grant2018hierarchical}. Given a  quantum circuit $V_L$ with a fixed depth  $L$ (say, a HWE quantum circuit), one samples states $\{\ket{\psi_i}\}$ from a given probability distribution, and sends them through the quantum circuit. Thus, the dataset is of the form $\{V_L\ket{\psi_i},y_i=L\}$. This can alternatively be thought of as a process where one  allows a system to evolve with a fixed, but possibly unknown process, and learns the length of the time evolution.

In Fig.~\ref{fig:depth_learning_ansatz} we show the HWE quantum circuit used to generate the dataset. We note that for each state in the dataset, one randomly samples the parameters in the ansatz. The main advantage of this type of architecture is that it can be readily implemented with a quantum circuit~\cite{kandala2017hardware}, and requires no previous training to generate the dataset. 

\subsubsection{Ansatz Depth Classification}

\begin{figure}[t]
    \centering
    \includegraphics[width=.8\columnwidth]{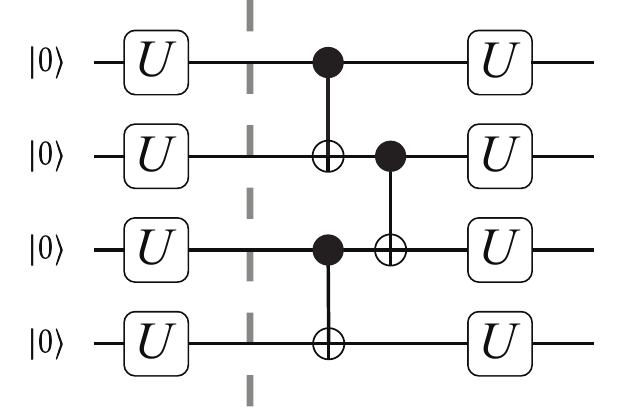}
    \caption{\textbf{Hardware efficient ansatz for depth learning.} The input states are sampled from the distribution of product states. As shown in the figure, we first initialize all qubits to the all zero state and apply a random single qubit unitary. Then, the final state is formed by sending the product state through a HWE ansatz,  where a single layer is composed of a CNOTs acting on alternating pars of qubits followed by local unitaries. Shown is the case of a circuit with one layer. }
    \label{fig:depth_learning_ansatz}
\end{figure}

\begin{table}[t]
\centering
\begin{tabular}{||c c c c||}
    \hline
    \# Qubits & Depths & Accuracy (\%) & Train/Test Set Size\\
    \hline \hline
    4 & 1 v 6 & 84.5 & 420/180\\
    \hline
    4 & 1 v 6 & 92 & 560/240\\
    \hline
    4 & 1 v 3 & 76.5 & 560/240\\
    \hline
    8 & 2 v 6 & 91 & 700/300\\
    \hline
    8 & 2 v 6 & 96.5 & 4000/1000\\
    \hline
    8 & 2 v 8 & 95.5 & 700/300\\
    \hline
    8 & 4 v 8  & 64 & 700/300\\
    \hline
    8 & 4 v 8 & 79 & 1400/600\\
    \hline
    8 & [1,2,3] v [5,6,7] & 93 & 3000/900\\
    \hline
\end{tabular}
\caption{\textbf{Testing classification accuracy of QCNNs trained to determine ansatz depth}.  The training and testing sets were generated using a $70\%/30\%$ testing/validation split. Here, the last row corresponds to the learning task where the states generated by depths $L=1,2,3$ where assigned the first label, while the states generated by $L=5,6,7$ where assigned the second label.}
\label{table:depth}
\end{table}

Here we consider the supervised machine learning task of binary classifying states in the depth-based dataset according to the number of layers they were created with. We note that this task has already been considered in~\cite{grant2018hierarchical}. Therein the authors sent a single copy ($m=1$) of the states into a classifying QNN, and achieved a classification accuracy of around $50\%$, which is no better than random guessing. Similar to the previous implementations on the NTangled dataset, we here also use a QCNN which takes as input $m=2$ copies of each state in the dataset. 

In Table~\ref{table:depth} we show the classification accuracy for classifying states according to their depth. Here we trained  10 trained models and reported the ones that achieved the smallest errors. We first consider a task where one set of states corresponds to a low-depth circuit, while the other set corresponds to a higher-depth circuit. For $4$-qubit problems, we can reach a classification accuracy of up to $92\%$ when one classifies states generated with HWE circuits with depths $L=1$ and $L=6$. Then, for $8$ qubits we can reach a classification accuracy of up to $96.5\%$ accuracy when classifying states generated with HWE circuits with depths $L=2$ and $L=8$.

For $8$ qubits, we also considered a binary classification task where multiple depths are binned together and assigned the same label. For example, we train a QCNN to distinguish between $8$ qubit states generated with depths $[1,2,3]$ versus states generated with depths $[5,6,7]$. In this case we obtain a classification accuracy of $93\%$. 

Here we can see that the performance in classifying states in a depth-based dataset can be greatly improved by sending $m=2$ copies of each state into the QCNN. The underlying reason is that, as previously mentioned, the depth-based dataset is also an entanglement-based dataset: states generated by low/high depth circuits have different entanglement properties. In the next section, we analyze the states in the depth-based dataset  to confirm this claim.

\subsubsection{Entanglement analysis of the depth-based  dataset}

In this section we analyze certain entanglement properties of the depth-based dataset. First, let us note that one can predict certain entanglement properties by analyzing the architecture of the ansatz. For instance, the number of layers is directly related to the minimum distance over which qubits can be correlated. Specifically,  one layer grows the light-cone for each qubit by two registers. Thus, we expect one can find full multipartite entanglement only when $L\geq \lfloor n/2 \rfloor$. Moreover, since the ansatz does not wrap around (i.e. the first qubit is not connected to the last qubit), we can expect that for small depths, the entanglement will concentrate in the middle qubits.

One can test these conjecture by sampling states from the dataset and computing the Concurrence~\cite{wootters2001entanglement} between all $\binom{n}{2}$ pairs of qubits. Here we recall that the Concurrence is a bipartite entanglement measure. For a mixed two-qubit state $\rho$, the Concurrence is defined as $\mathscr{C}(\rho) = \max \{0,\lambda_1 - \lambda_2 - \lambda_3 - \lambda_4\}$ where $\{\lambda_i\}$ are the eigenvalues (from greatest to smallest) of $\tilde{\rho} = (\sigma_y\otimes\sigma_y)\rho^*(\sigma_y\otimes\sigma_y)$.

\begin{figure}[t]
    \centering
    \includegraphics[width=\columnwidth]{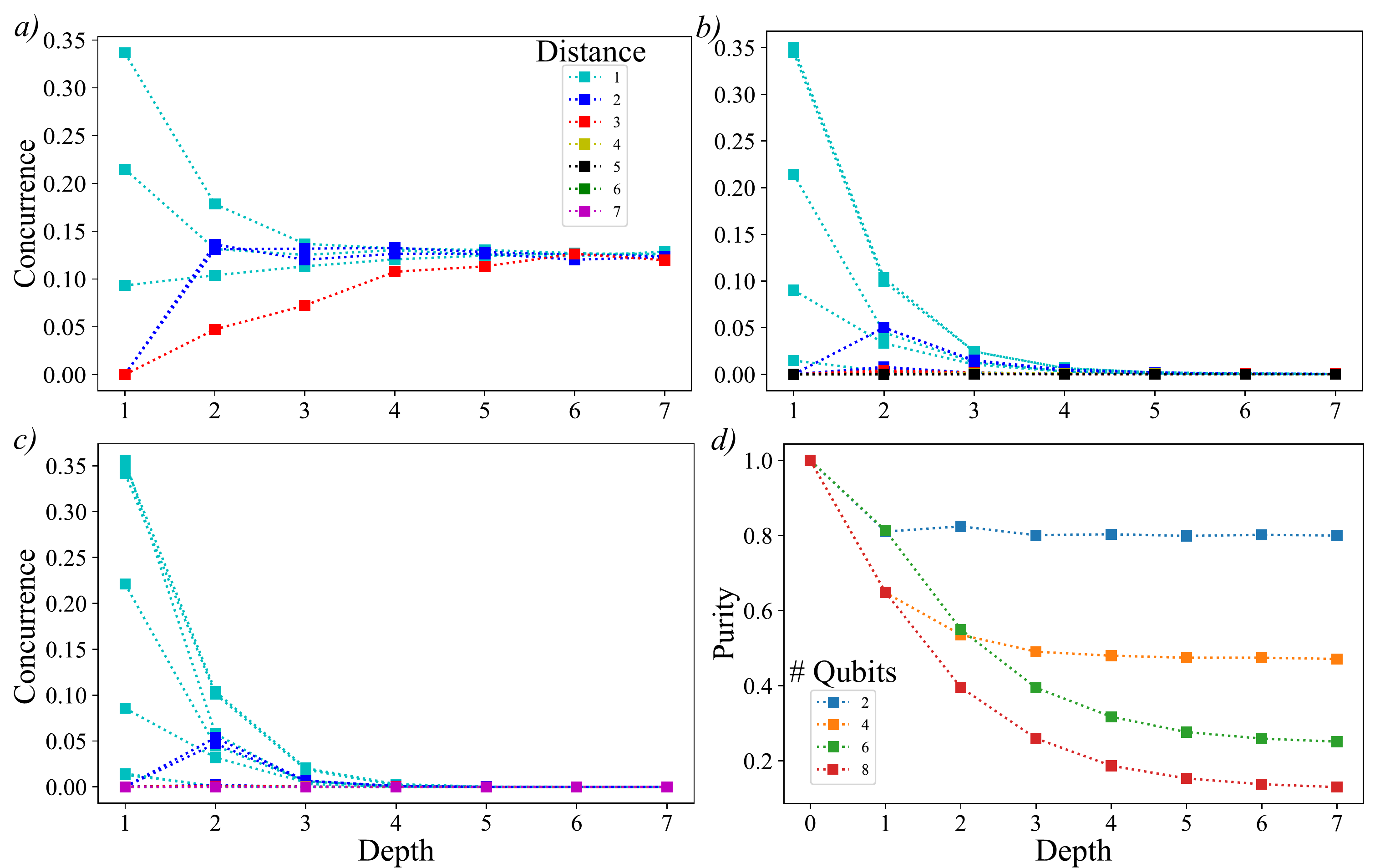}
    \caption{\textbf{Entanglement properties of the states in the depth-based dataset.} (a-c) Concurrence as a function of the circuit depth for system respectively composed of $n=4,6,8$ qubits. Here, lines of the same color correspond to the Concurrence between qubits at the same distance (see text). (d) Averaged purity of the reduced states of the first $n/2$ qubits. In all cases, the averages were obtained by sampling  $1000$ (for 4 qubits), $2000$ (for 6 qubits) and $4000$ (for 8 qubits) output states from fixed depth circuits.}
    \label{fig:depth_entanglement}
\end{figure}

In Fig.~\ref{fig:depth_entanglement} (a-c) we plot the average Concurrence as a function of the circuit depth for datasets in $n=4,6,8$ qubits. Here we compute the Concurrences of all two-qubit reduced states $\rho_{ij}$, where the indexes $i$  and $j$ denote the considered qubits. For an easy comparison, we use the same color for the Concurrence of qubits at the same distance $l$, i.e., we use the came color for the Concurrence of all reduced states $\rho_{ij}$ with $|i-j|=l$.  Here we see that neighbouring qubits are indeed more entangled (have a higher Concurrence) at low ansatz depth. Then, as the circuit depth increases, all qubits become entangled among themselves, and due to the monogamy of entanglement~\cite{cornelio2013multipartite}, the Concurrences converge to a fixed value independent of the distance $l$. 

While the Concurrence is a measure of how much every pair of qubits are entangled amongst themselves, one can also consider a more coarse grained entanglement measure. Namely, we here also compute the purity $\Tr[\rho_{[n/2]}^2]$, where $\rho_{[n/2]}^2$ denotes the reduced state of the first $n/2$ qubits. This quantity  measures the bipartite entanglement between the first, and the second half of the qubits. We show results in Fig.~\ref{fig:depth_entanglement}(d) for datasets of different number of qubits $n$. Here we can see that as the circuit depth depth increases, the purities achieve the saturation value of $\frac{2^{n+1}}{1+2^{2n}}$, which corresponds to the Haar-averaged reduced state purity. As expected this is due to the fact that deep HWE ansatzes form unitary $2$-designs~\cite{brandao2016local,harrow2018approximate}.

\section{Discussion}

In this paper we introduced the NTangled quantum dataset for Quantum Machine Learning (QML). The states in the NTangled dataset have different amounts and types of multipartite entanglement. As we show here, the NTangled datasets can be used to benchmark QML models for supervised learning tasks such as binary classification. 

The main advantage of using the NTangled dataset to benchmark QML architectures is that the data is already in quantum form and requires no embedding scheme. This is in contrast to using classical data (such as the MNIST or Iris datasets), where one first needs to encode, or embed, the classical data in the quantum Hilbert space. As we argued here, using classical data has two main disadvantages. First, it is still unclear what makes a good embedding scheme, and it has been shown that the choice of embedding can actually negatively impact the model's trainability and lead to issues such as barren plateaus. Second, the utility of QML with classical problems is questionable, as it is still unclear if one can achieve a quantum advantage using classical data. These two issues can be avoided by using datasets composed of true quantum data, such as our NTangled dataset.

Our first contribution was to train a generator Quantum Neural Network (QNN) to create the states in the NTangled dataset. Specifically, we created datasets on $n=3,4,8$ and $12$ qubits of states with different amounts of multipartite entanglement. For this purpose we benchmarked different ansatzes for the generator QNN. Then, we reported how sampling from different sets of input states can affect the success probability that the states generated from the QNN have a desired value of entanglement. In most cases we achieved a success probability of over $95\%$, with some architectures even being able to achieve success with  $100\%$ probability. In the Appendix and in the Github repository of~\cite{schatzki2021github} we give a detailed description of how to employ our trained generator QNNs to use the NTangled datasets. The datasets can be efficiently stored, as they come in the form of a parametrized quantum circuit plus the values of the trained parameters.

We then showcased using the NTangled dataset for the supervised learning task of binary classification. Here, we trained a classifying convolutional QNN to classify states according to their entanglement level. We showed how the number of copies of each state that are input to the classifying QNN can affect its performance. For inputting two copies of each state, we achieved a classification with testing accuracy of over $90\%$.

In addition, we also studied an alternative entanglement-based dataset, where states are generated by sampling the output of quantum circuits with different depths. Here, we showed that we can train a QNN to classify states on $n=4,8$ qubits according to the depth for which they were created (shallow vs deep circuits) with an accuracy of up to $95\%$. The main advantage of this dataset is that it is easy to implement, and requires no training for the generator QNN.

Our work paves the way towards large-scale non-trivial quantum datasets. While some other quantum datasets have been considered, these can have some crucial issues. For instance, some proposals use datasets formed by ground-states of the one-dimensional Ising model~\cite{uvarov2020machine,banchi2021generalization}. However, these can be trivially classified  by performing a single-qubit measurement and do not really require a QML model. In addition, while some small scale publicly available datasets have started to appear, there is still need for large-scale datasets that can be used to benchmark the scaling of QML models and guarantee that they will not have trainability issues such as barren plateaus~\cite{mcclean2018barren,cerezo2020cost}.

We finally note that our work also has implications for using QNNs to study multipartite states and for quantum entanglement theory. Usually, generating states with the same entanglement properties is achieved by  using local operations and classical communication (LOCC). That is, one takes a state with a desired amount of entanglement and applies to it operations in LOCC. In this sense, our approach for generating the NTangled dataset is fundamentally different. Here, we took as input a set of states that have the same amount of entanglement, and we trained a global unitary (which is not in LOCC), so that its output states have a desired amount of entanglement. This procedure then allowed us to generate entangled states that do not belong to the same LOCC orbits. Further, our results suggest that this method is scalable and can be tuned to generate new distributions with desired entanglement properties.

Thus, we expect that our NTangled dataset will be relevant for both the quantum computing and the quantum information communities.

\section{Acknowledgments}
We thank Elliott Ball and Lukasz Cincio for useful discussions, Linjian Ma for providing code for complex CP decomposition, and Jacob Beckey and Guangkuo Liu for discussions and improvements on Proposition~\ref{theo:1}.  L.S. was supported by the U.S. DOE through a quantum computing program sponsored by the Los Alamos National Laboratory (LANL) Information Science \& Technology Institute. L.S. also acknowledges support from the NSF Quantum Leap Challenge Institute for Hybrid Quantum Architectures and Networks (NSF Award 2016136). AA, PJC, and MC were initially supported by Laboratory Directed Research and Development (LDRD) program of LANL under project number 20190065DR. PJC also acknowledges support from the LANL ASC Beyond Moore's Law project. MC also acknowledges support from the Center for Nonlinear Studies at LANL.  This work was supported by the U.S. DOE, Office of Science, Office of Advanced Scientific Computing Research, under the Accelerated Research in Quantum Computing (ARQC) program.

\bibliography{references,quantum}

\onecolumngrid

\pagebreak

\setcounter{section}{0}
\setcounter{proposition}{0}
\setcounter{theorem}{0}
\setcounter{corollary}{0}

\section*{Appendix}

Here we present additional details for the main results in our manuscript. In Appendix~\ref{sec:appendix_I} we further discuss the theory of multipartite entanglement, and we present how the  Concentratable Entanglement (CE) can be computed. Here we also present the proofs for our Theorems. Then, in Appendix~\ref{sec:appendix_II} we describe the ansatzes  employed throughout this work to generate the NTangled datasets. Finally, in Appendix~\ref{sec:appendix_III} we present additional numerical results to characterize the states in the $3$ and $4$ qubit NTangled datasets. Namely, the distributions of CE for generators of various depths.

\section{Multipartite Entanglement and Concentratable Entanglement Properties}\label{sec:appendix_I}

\subsection{Multipartite Entanglement}
There exist (among others) two ways of considering multipartite entanglement: local operations and classical communication (LOCC) and stochastic local operations and classic communication (SLOCC). SLOCC often holds more operational significance due to probabilistic protocols being more experimentally accessible then purely deterministic. Two states are considered to be SLOCC equivalent if $\ket{\psi} = A_1\otimes A_2\otimes\ldots\otimes A_n \ket{\phi}$, where $\{A_i\}$ are invertible operators. For $3$ qubits, there exist six non-equivalent SLOCC classes, two of which represent truly multipartite entangled states \cite{dur2000three}. For any larger systems, there are an infinite number. In $4$ qubits, for example, SLOCC classes can be described via nine canonical forms, each crucially with free parameters \cite{li2009slocc}. Thus, generating a variety of entangled states presents a difficult task in general.

The $n$-tangle is one way to quantify entanglement in multipartite states \cite{wong2001potential}. This is defined for even numbers of qubits and takes, for pure states, the form: 
\begin{align}
    \tau_n(\ket{\psi}) = |\braket{\psi}{\tilde{\psi}}|^2,\ 
    \ket{\tilde{\psi}} := \sigma_y^{\otimes n}\ket{\psi^*}.
\end{align}
However, biseparable states, such as $\ket{\Psi^-}^{\otimes 2}$, obtain the maximal value of 1.

True multipartite entanglement can be analyzed in various methods, such as invariant polynomials~\cite{szalay2015multipartite}. Here, for further reference, we will briefly discuss the Schmidt measure~\cite{eisert2001schmidt}. Any state can be written in the form $\ket{\psi} = \sum_{i=1}^ra_i\ket{\psi_1^{(i)}}\otimes\ket{\psi_2^{(2)}}\otimes\ldots\otimes\ket{\psi_n^{(i)}}$. Taking r to be the minimal number of terms in such a decomposition for a given state, the Schmidt measure is $\log_2r$. Note that this is equivalent to tensor CP decomposition~\cite{kolda2009tensor}. For example, a GHZ state on any system size is rank 2 and thus has Schmidt measure 1. This measure has operational relevance in that a state cannot be converted via SLOCC protocols into one of higher Schmidt measure~\cite{chitambar2008tripartite}. 

In this work we considered the Concentratable Entanglement (CE) measure defined in~\cite{beckey2021computable}. The Concentratable Entanglement is given by
\begin{equation}
    \mathcal{C}(\ket{\psi})= 1-\frac{1}{2^n}\sum_{\alpha\in Q}\Tr[\rho_\alpha^2],
\end{equation}
where $Q$ is the power set of $[n]$ and $\rho_\alpha$ is the reduced state of subsystems in index subset $\alpha$. Experimentally, CE is accessible through sending two copies of a state into the controlled Swap Test in Fig.~\ref{fig:swap_test}. Simply measure the probability of obtaining all zero's on the ancilla registers: $\mathcal{C}(\ket{\psi}) = 1-p(\vec{0})$. Further, the $n$-tangle can be easily accessed through the probability of the all one measurement as $p(\vec{1}) = \frac{\tau_n(\ket{\psi})}{2^n}$. Lastly, note that GHZ states follow $\mathcal{C}(\ket{\text{GHZ}_n}) = \frac{1}{2}-\frac{1}{2^n}$.

\subsection{Proof of Theorem 1}

Let us recall Theorem~\ref{theo:1SM}: 
\begin{theorem}\label{theo:1SM}
    Given two $n$-qubit pure states $\psi=\dya{\psi}, \phi=\dya{\phi}$, the difference between their CE follows
    \begin{align}
        |C(\ket{\psi})-C(\ket{\phi})| &\leq \sqrt{2}(2-2^{2-n})D_{tr}(\psi,\phi)\\
        & \leq 2\sqrt{2}D_{tr}(\psi,\phi)\,,
        \label{eq:CE_TheoremSM}
    \end{align}
where, for matrices $A$ and $B$, $D_{tr}(A,B)=\frac{1}{2}\Tr[\sqrt{(A-B)\ad (A-B}]$ denotes the trace distance.
\end{theorem}

\begin{proof}
    Here by $\psi_s$ we denote the reduced density matrix for index subset $s\in [n]$.
    \begin{align*}
        |\Tr[{\psi_s^2-\phi_s^2}]| & = |\Tr{[(\psi_s+\phi_s)(\psi_s-\phi_s)]}|\\
                                   & \leq \norm{\psi_s+\phi_s}_2\norm{\psi_s-\phi_s}_2\\
                                   & \leq (\norm{\psi_s}_2+\norm{\phi_s}_2)\norm{\psi_s-\phi_s}_2\\
                                   & \leq 2\norm{\psi_s-\phi_s}_2\\
                                   & \leq {\sqrt{2}}\norm{\psi_s-\phi_s}_1\\
                                   & \leq {\sqrt{2}}\norm{\psi-\phi}_1.
    \end{align*}
    Where the first inequality comes from Cauchy-Schwarz, the second/third from triangle inequality (and the maximum H.S. norm for density matrices), {the fourth from upper bounds on 2-norm given in} \cite{Coles2012_decoherence,coles2019strong}, and the last from trace norm monotonicity under CPTP maps. Recall the definition of Concentratable Entanglement:
    
    \[
         C_{\ket{\psi}}([n]) = 1-\frac{2}{2^n}-\frac{1}{2^n}\sum_{k=1}^{n-1}\sum_{|s|=k} \Tr[\psi_s^2]\,.
    \]
    Applying the bound above on the $2^n-2$ purities in the summation and noting that $\norm{A-B}_1=2D_{tr}(A,B)$, we arrive at
    
    \begin{align}
        |C_{\ket{\psi}}([n])-C_{\ket{\phi}}([n])|& \leq {\sqrt{2}}*\frac{2^n-2}{2^n} \norm{\psi-\phi}_1\nonumber\\
        &= {\sqrt{2}(2-2^{2-n})}D_{tr}(\psi,\phi)\,.
    \end{align}
\end{proof}

\begin{figure}
    \centering
    \includegraphics[scale=0.6]{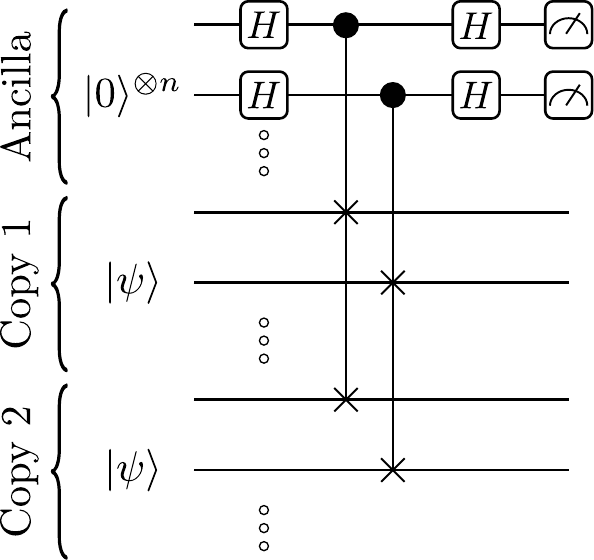}
    \caption{\textbf{Swap test circuit for measuring the Concentratable Entanglement of a state.} $\mathcal{C}(\ket{\psi}) = 1-p(\vec{0})$, where $p(\vec{0})$ is the probability of measuring all  the ancilla registers in the zero state.}
    \label{fig:swap_test}
\end{figure}

\subsection{Proof of Theorem 2}

\begin{theorem}\label{thm:3q_ent_appendix}
    Any 3 qubit state $\ket{\psi}$ with $\mathcal{C}(\ket{\psi}) > 0.25$ has true multipartite entanglement.
\end{theorem}

\begin{proof}
It is obvious that a product state will have CE of  0. Thus, we consider a state that is entangled but not multipartite entangled, i.e. without loss of generality taking the form $\ket{\psi} = \ket{AB}\ket{C} = \ket{\phi}\ket{C}$, where $\ket{\phi}=\ket{AB}$. Using the fact that $\Tr[\rho_s^2] = \Tr[\rho_{\bar{s}}^2]$, where $\bar{s}$ is the complement of index subset $s$, we write the CE of this state as:

$$
    \mathcal{C}(\ket{\psi}) = 1-\frac{1}{4}(1+\Tr[\psi_a^2] + \Tr[\psi_b^2] +\Tr[\psi_c^2]).
$$
Of course, $\Tr[\psi_c^2]$ = $\Tr[(\dya{C})^2] = 1$. Also note that $\psi_a = \Tr_b[\dya{AB}] = \phi_a$. A similar property holds for $\psi_b$. Combining,

\begin{align*}
    \mathcal{C}(\ket{\psi}) & = 1-\frac{1}{4}(2+\Tr[\psi_a^2] + \Tr[\psi_b^2])\\
    & = 1-\frac{1}{4}(2+\Tr[\phi_a^2] + \Tr[\phi_b^2]) \\ 
    & = \mathcal{C}(\ket{\phi}).
\end{align*}

It is easy to see that the maximal value of CE for 2 qubits is $0.25$ by noting that $\mathcal{C}(\ket{\phi}) = \frac{1}{2}-\frac{1}{2}\Tr[\phi_a^2]$ and $\Tr[\phi_a^2] \in [0.5, 1]$.
    
\end{proof}

\section{Ansatz and Classifier Implementation} \label{sec:appendix_II}
Here we describe our ansatzes in greater detail to facilitate replication of results and usage of the datasets. We denote an arbitrary single qubit unitary by $U3(i;\alpha,\beta,\gamma) = e^{i(\beta+\gamma)/2}R_Z(\beta)R_Y(\alpha)R_Z(\gamma)$ where the $i$-index indicates the unitary acts on qubit i. $CNOT(i,j)$ indicates a controlled not with register $i$ as the control and $j$ as the target. Similarly, $CZ(i,j)$ indicates a controlled Z upon registers i and j. Lastly, we denote an arbitrary two qubit unitary by $2QU(i,j;\vec{\theta})$, where $\vec{\theta} \in [0,2\pi)^{15}$ and i and j indicate the two qubits being acted upon. In Algorithm~\ref{alg:two_qubit_unitary} we describe the circuit for enacting $2QU$~\cite{vatan2004optimal} (also depicted in Fig.~\ref{fig:two_qubit_unitary}):

In Algorithms~\ref{alg:hwe_generator},~\ref{alg:sea_generator}, and ~\ref{alg:conv_generator} we respectively describe the pseudocode to generate the structure of the HWE, SEA and CONV ansatzes for the generator QNN in the NTangled dataset. Finally, in Algorithm~\ref{alg:depth_learning_ansatz} we present the psuedocode for the ansatz employed to generate the depth-based dataset.

\begin{algorithm}\label{alg:two_qubit_unitary}
\DontPrintSemicolon
\KwIn{Two quantum registers i and j; $\vec{\theta} \in [0,2\pi)^{15}$ }
    $U3(i; \vec{\theta}_{0:3})$;
    
    $U3(j; \vec{\theta}_{3:6})$;
    
    $CNOT(j,i)4$;
    
    $R_Z(i; \theta_7)$;
    
    $R_Y(j; \theta_8)$;
    
    $CNOT(i,j)$;
    
    $R_Y(j; \theta_9)$;
    
    $CNOT(j,i)$;
    
    $U3(i; \vec{\theta}_{10:12})$;
    
    $U3(j; \vec{\theta}_{12:15})$;
\caption{\sc Arbitrary Two Qubit Unitary (2QU)}
\end{algorithm}

\begin{figure}
    \centering
    \includegraphics[width=0.5\columnwidth]{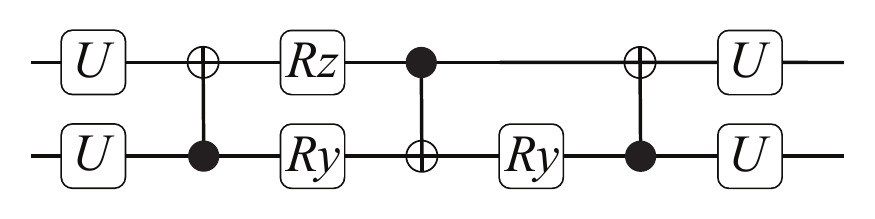}
    \caption{\textbf{Decomposition for an arbitrary two qubit unitary.}}
    \label{fig:two_qubit_unitary}
\end{figure}

\begin{algorithm}[]\label{alg:hwe_generator}
\DontPrintSemicolon
\KwIn{Quantum product state $\ket{\psi^{(\tin)}}=\bigotimes_i \ket{\psi_i^{(\tin)}}$; generator weights $\vec{\theta}: \theta_{a,b,c}\in [0,2\pi)$ in a tensor of shape (L,n,3), where L is the depth and n is the number of qubits.}
\KwOut{An entangled quantum state $\ket{\psi^{(\tout)}} = U_{hwe}(\vec{\theta})\ket{\psi^{(\tin)}}$.}
\For{d=0; d<L; ++d}{
    \For{i=0; i < n; ++i}{
        $U3(i;\vec{\theta}_{d,i,0:3})$
    }

    \For {i=0; $i < \lfloor n/2 \rfloor$ ; ++i}{
        $CZ(2i, 2i+1)$;
    }
    
    \For{i=0; i < n; ++i}{
        $U3(i;\vec{\theta}_{d,i,0:3})$
    }
    
    \For {i=0; $i < \lfloor *n-1)/2 \rfloor$ ; ++i}{
        $CZ(2i+1, 2i+2)$;
    }
}
    \caption{\sc Hardware Efficient Ansatz State Generator}
\end{algorithm}

\begin{algorithm}[h]\label{alg:sea_generator}
\DontPrintSemicolon
\KwIn{Quantum product state $\ket{\psi^{(\tin)}}=\bigotimes_i \ket{\psi_i^{(\tin)}}$; generator weights $\vec{\theta} : \theta_{a,b,c}\in [0,2\pi)$ in a tensor of shape (L,n,3), where L is the depth and n is the number of qubits.}
\KwOut{An entangled quantum state $\ket{\psi^{(\tout)}} = U_{sea}(\vec{\theta})\ket{\psi^{(\tin)}}$.}
\For{d=0; d<L; ++d}{
    \For{i=0; i < n; ++i}{
        $U3(i;\vec{\theta}_{d,i,0:3})$
    }

    \For {i=0; $i \leq n$; ++i}{
        $CNOT(2i, (2i+1)\% n)$;
    }
}
    \caption{\sc Strongly Entangling Ansatz State Generator}
\end{algorithm}

\begin{algorithm}[h]\label{alg:conv_generator}
\DontPrintSemicolon
\KwIn{Quantum product state $\ket{\psi^{(\tin)}}=\bigotimes_i \ket{\psi_i^{(\tin)}}$; generator weights $\vec{\theta} : \theta_{a,b,c}\in [0,2\pi)$ in a tensor of shape (L,30*$\lfloor n/2 \rfloor$), where L is the depth and n is the number of qubits.}
\KwOut{An entangled quantum state $\ket{\psi^{(\tout)}} = U_{conv}(\vec{\theta})\ket{\psi^{(\tin)}}$.}

\For{d=0; d < L; ++d}{
    j = 0;
    
    \For{i=0; $i < \lfloor n/2 \rfloor$; ++q}{
        $2QU(2i,2i+1;\vec{\theta}_{d,j:j+15})$;\newline
        
        j += 15;
    }
    
    \For{i=0; $i < \lfloor n/2 \rfloor$; ++q}{
        $2QU(2i+1,(2i+2)\% n ;\vec{\theta}{d,j:j+15})$;\newline
        
        j += 15;
    }
    
}
    \caption{\sc Convolutional Ansatz for State Generation}
\end{algorithm}

\begin{algorithm}[h]\label{alg:depth_learning_ansatz}
\DontPrintSemicolon
\KwIn{Quantum zero state $\ket{\vec{0}}$; parameters $\vec{\theta} : \theta_{a,b,c}\in [0,2\pi)$ in array of shape (L+1,n,3), where L is the depth and n is the number of qubits.}
\KwOut{A quantum state $\ket{\psi^{(\tout)}} = V(\theta) \ket{\psi^{(\tin)}}$}
\For{i=0; i < n; ++i}{
    $U3(i;\vec{\theta_{0,i,0:3}})$
}

\For{d=1; d < L+1; ++d}{
    \For {i=0; $i < \lfloor n/2 \rfloor$ ; ++i}{
        $CNOT(2i, 2i+1)$;
    }
    
    \For {i=0; $i < \lfloor (n-1)/2 \rfloor$ ; ++i}{
        $CNOT(2i+1, 2i+2)$;
    }
    
    \For {i=0; i<n ; ++i}{
        $U3(i;\vec{\theta}_{d,i,0:3})$
    }
}
    \caption{\sc Depth Learning Hardware Efficient Ansatz}
\end{algorithm}
\clearpage

\section{Additional numerical results} \label{sec:appendix_III}
In Fig.~\ref{fig:entanglement_distributions} (following page) we present the distribution of entanglement for the $3$ and $4$ qubit NTangled dataset. Note that the testing set here for all models was 500 randomly sampled product states. This figure demonstrates that for $4$ qubits, the SEA ansatz had trouble generating anything but highly entangled states.

\begin{sidewaysfigure*}[h]
    \centering
     \includegraphics[height=0.65\textheight]{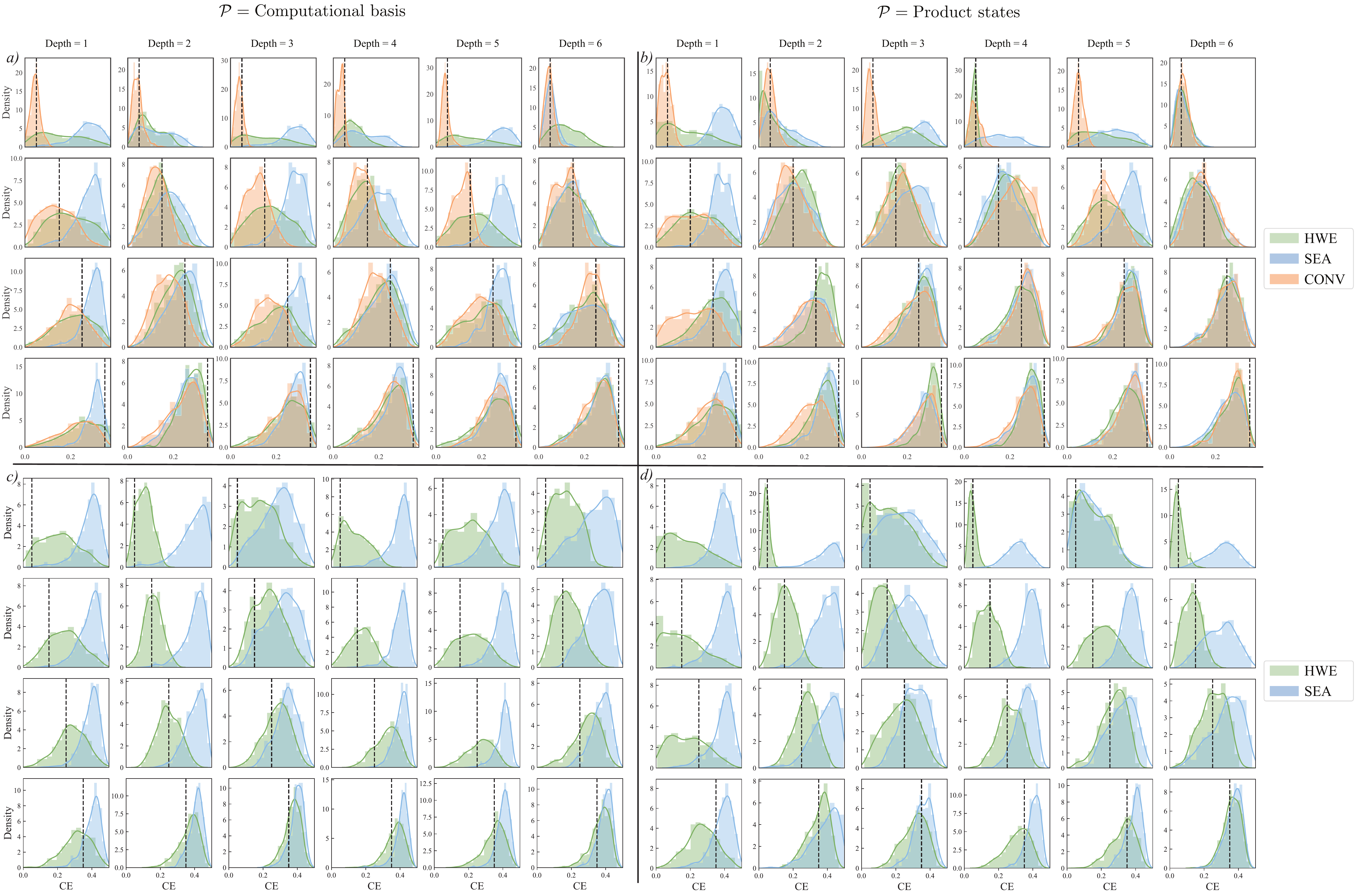}
     \caption{ \textbf{Distribution of entanglement for the $3$ and $4$ qubit NTangled dataset.}  The left (right) column displays generators trained on the computational basis (random product states). (a \& b) are 3 qubit generators while (c \& d) are 4 qubit generators. {Generators were trained for $\zeta\in\{0.05,0.15,0.25,0.35\}$. The data shown comes from the best of 50 trained models for each task.}}
    \label{fig:entanglement_distributions}
\end{sidewaysfigure*}

\end{document}